\documentclass[12pt, draftclsnofoot, onecolumn]{IEEEtran}

\usepackage{amsfonts}
\usepackage{slashbox}
\usepackage{amsmath}
\usepackage{multirow}
\usepackage{graphicx}
\usepackage{amsfonts}
\usepackage{amsmath,epsfig}
\usepackage{mathbbold}
\usepackage{stfloats}
\usepackage{amssymb}
\usepackage{url}
\usepackage{bm}
\usepackage{cite}
\usepackage{amsmath}
\usepackage{amssymb}
\usepackage{latexsym}
\usepackage{multirow}
\usepackage{epsfig}
\usepackage{graphics}
\usepackage{mathrsfs}
\usepackage{algorithmic}
\usepackage{algorithm}
\usepackage{subfigure}
\usepackage{slashbox}
\usepackage[all]{xy}

\usepackage{pifont}
\usepackage{bbding}
\usepackage{stmaryrd}
\usepackage{amssymb}
\usepackage{amsfonts}
\usepackage{epic}
\usepackage{stfloats}
\usepackage{epstopdf}
\usepackage{bm}
\usepackage{xcolor}
\usepackage{mathrsfs}
\usepackage{pifont}
\usepackage{bbding}
\usepackage{amsmath,epsfig}
\usepackage{mathbbold}
\usepackage{stmaryrd}
\usepackage{amsmath}
\usepackage{amssymb}
\usepackage{amsthm}

\usepackage{amsfonts}
\usepackage{epic}
\usepackage{graphicx}
\usepackage{subfigure}
\usepackage{enumerate}
\usepackage{latexsym}
\usepackage{epstopdf}
\usepackage{multirow}
\usepackage{stfloats}
\usepackage{bm}
\usepackage{booktabs}
\usepackage{color}
\usepackage{cases}
\usepackage{array}
\usepackage{makecell}
\usepackage{times}

\usepackage{diagbox}
\usepackage{slashbox}
\usepackage{booktabs}

\newtheorem{remark}{Remark}

\newtheorem{lemma}{Lemma}

\newcommand{\mv}[1]{\mbox{\boldmath{$ #1 $}}}

\newcommand{\A}{\bm A}
\newcommand{\w}{\bm w}

\newcommand{\R}{\bm R}

\newcommand{\I}{\bm{I}}

\newcommand{\bb}{\bm{b}}
\newcommand{\N}{\mathcal{N}}

\newcommand{\II}{\mathcal{I}}

\newcommand{\F}{\bm F}

\newcommand{\Ss}{\bm S}

\newcommand{\uuu}{\bm u}
\newcommand{\dd}{\bm d}

\newcommand{\h}{\bm h}
\newcommand{\g}{\bm g}

\newcommand{\ttheta}{\mathbf \Theta}

\newcommand{\K}{\mathcal{K}}

\newcommand{\E}{\mathcal{E}}

\graphicspath{{./figs/}}

\begin{document}
\title{ Joint Active and Passive   Beamforming Optimization for  Intelligent Reflecting Surface Assisted SWIPT under QoS Constraints}
\author{\IEEEauthorblockN{Qingqing Wu,  \emph{Member, IEEE} and Rui Zhang, \emph{Fellow, IEEE}
\thanks{ The authors are with the Department of Electrical and Computer Engineering, National University of Singapore, email:\{elewuqq, elezhang\}@nus.edu.sg.}
   } }

\maketitle
\vspace{-0.8cm}
\begin{abstract}
Intelligent reflecting surface (IRS) is a new and revolutionizing technology for achieving spectrum and energy efficient wireless networks. By leveraging massive low-cost passive elements that are able to reflect radio-frequency (RF) signals with adjustable phase shifts, IRS can achieve high passive beamforming gains, which are particularly appealing for improving the efficiency of RF-based wireless power transfer.   Motivated by the above, we study in the paper an IRS-assisted simultaneous wireless information and power transfer (SWIPT) system. Specifically, a set of  IRSs are deployed to assist in the information/power transfer from a multi-antenna access point (AP) to multiple single-antenna information users (IUs) and energy users (EUs), respectively.   We aim to minimize the transmit power at the AP  via jointly optimizing its transmit precoders and the reflect phase shifts at all IRSs, subject to the quality-of-service (QoS) constraints at all users, namely, the individual  signal-to-interference-plus-noise ratio (SINR) constraints at IUs and energy harvesting constraints at EUs. However, this optimization problem is non-convex with intricately coupled variables, for which the existing  alternating optimization approach is shown to be inefficient as the number of QoS constraints increases.
To tackle this challenge,  we first apply proper transformations on the QoS constraints and then propose an efficient iterative algorithm by applying the penalty-based method.
Moreover, by exploiting the short-range coverage of IRSs, we further propose a low-complexity algorithm by optimizing  the phase shifts of all IRSs in parallel.  Simulation results demonstrate the effectiveness of  IRSs for enhancing the performance of  SWIPT systems as well as the  significant performance gains achieved by our proposed algorithms over benchmark schemes. The impact of  IRS on the transmitter-receiver design for SWIPT is also unveiled.
\end{abstract}
\vspace{-5mm}

\section{Introduction}
The  number of Internet-of-Things (IoT) devices (e.g., electronic tablets, sensors, wearables, and so on) worldwide is anticipated to  skyrocket from about 7 billion in 2018 to 22 billion by 2025, laying the foundation of the future  smart home, city and nation. 
  Such a massive  number of wireless devices thus require a scalable solution for  providing them not only ubiquitous communication connectivity  but also perpetual energy supply in the future (say, the  fifth generation (5G) and beyond) wireless network. To this end, the dual use of radio frequency (RF) signals for simultaneous wireless information and power transfer (SWIPT) has recently gained an upsurge of interest \cite{zeng2017communications,clerckx2018fundamentals}. However, an  energy user (EU) typically requires much higher receive power than that of the signal for  an information user (IU), due to their  drastically  different  receiver sensitivities and application requirements in practice  \cite{clerckx2018fundamentals}.
 As such, the low efficiency of wireless power transfer (WPT)  for EUs over long distances has been considered as the performance bottleneck in practical SWIPT systems. Although the massive multiple-input multiple-output (MIMO) technology is able to improve the WPT efficiency considerably by leveraging the large array/beamforming gain at the WPT/SWIPT transmitter \cite{yang2015throughput,yang2018wireless}, the required high complexity, high energy consumption, and high hardware cost are still the main roadblocks to its implementation in practice, especially at the increasingly higher RF (e.g., millimeter wave) frequencies.

Recently, intelligent reflecting surface (IRS)  has been proposed as a promising cost-effective solution  to improve the wireless communication spectrum and energy efficiency  \cite{JR:wu2019IRSmaga,JR:wu2018IRS}. By dynamically adjusting  the phase shifts of the reflected signals via a vast number of low-cost passive elements based on the time-varying environment,    IRS can  achieve fine-grained three-dimensional (3D) passive beamforming gains and thereby reconfigure the wireless propagation channels to be favorable for communication performance optimization. Compared with the conventional active  beamforming/relaying via massive MIMO,  IRS eliminates signal amplification and regeneration, thus enjoying much lower hardware cost, energy consumption, and interference contamination.   As such,  IRSs feature effective short-range/local coverage  and can be densely  deployed with a scalable  cost, yet without the need of sophisticated interference management provided that they are  sufficiently separated from each other \cite{JR:wu2019IRSmaga}.  Furthermore, IRSs are  of low profile and can be practically fabricated to be conformal to mount on arbitrarily shaped surfaces to cater for different  application scenarios.  All these compelling  advantages have spurred a great deal of  interest recently in investigating and building   IRS or its various equivalents \cite{di2019smart,basar19_survey}. In January 2017, a European-funded pilot project ``VISORSURF'' was launched to build the prototype of a software-controlled meta-surface, with the ultimate goal of making the wireless radio propagation environment fully reconfigurable. In November 2018,  NTT DoCoMo and Metawave jointly conducted preliminary experimental tests, which showed that by properly deploying a meta-structure based reflect-array,  the communication quality can be greatly improved over 5G alone, with a range extension of about 35 meters.  To capitalize on this growing  opportunity,  new startup companies, e.g.,   Greenerwave   and Pivotal Commware, have appeared recently to focus on the commercialization of IRS-type technologies  for consumer-grade use cases.

The new research paradigm of IRS-aided wireless communication  has been extensively studied recently \cite{JR:wu2019IRSmaga,JR:wu2018IRS,cui2019secure,guan2019intelligent,chen2019intelligent,dongfang2019,fu2019intelligent,yang2019intelligent, han2018large,huangachievable,yan2019passive}.  In particular, \cite{JR:wu2019IRSmaga} provided a comprehensive overview of IRS-aided  wireless networks. Furthermore, it was shown in  \cite{JR:wu2018IRS} that via jointly optimizing the  active/transmit  and passive/reflect  beamforming in an IRS-aided wireless network, the signal-to-interference-plus-noise ratio (SINR)  performance of all users in the network can be significantly improved, regardless of whether they are aided directly by the IRS or not. The  joint active and passive beamforming design  was also investigated in  other system setups, e.g.,  physical layer security \cite{guan2019intelligent,cui2019secure,chen2019intelligent,dongfang2019}, orthogonal frequency division multiplexing  (OFDM) systems \cite{JR:yang2019irs,JR:beixiong2019irs}, and non-orthogonal multiple access  \cite{fu2019intelligent,yang2019intelligent}.  To implement IRS phase-shifts in practice, some recent works studied the use of IRS with discrete phase shifts \cite{JR:wu2019discreteIRS,huang2018energy} or non-linear reflection amplitude versus phase-shift  \cite{abeywickrama2019intelligent} due to practical hardware constraints.

While the above works  focus on exploiting IRS for enhancing the wireless communication or information transmission performance,  the high  beamforming gain achieved by passive IRS is also appealing for WPT \cite{JR:wu2019IRSmaga,mishra2019channel}. By leveraging  the intelligent reflection over their large aperture, IRSs can help compensate the high RF signal attenuation over long distance and thereby establish effective energy harvesting/charging zones for hot-spot areas  in their proximity,  as illustrated in Fig. 1.  This is of great practical  significance for efficiently  extending the coverage of WPT and realizing the envisioned battery-free IoT networks in the future. To reap this benefit,  the weighted sum-power and sum-rate optimization problems in IRS-aided SWIPT systems were recently studied in \cite{wu2019weighted} and \cite{pan2019intelligent}, respectively. Although the designs in the above works can deal with the performance fairness issue among the users to a certain extent by e.g., adjusting their weights in the correspondingly formulated optimization problems, it remains undressed how  a given set of user quality-of-service (QoS) requirements on the individual SINRs  for IUs and the individual harvested energy amounts for EUs can be efficiently achieved. Furthermore, from the optimization perspective,  the prior works (e.g., \cite{cui2019secure,guan2019intelligent,chen2019intelligent,dongfang2019,fu2019intelligent,huangachievable,yang2019intelligent,huang2018energy,wu2019weighted,pan2019intelligent}) all adopted an alternating optimization approach by successively optimizing the transmit precoders and the IRS's phase shifts in an iterative manner, which, however,  becomes highly inefficient  as the number of QoS constraints increases. This is due to the fact  that by fixing a subset of the optimization variables (e.g.,  transmit precoders), the feasible set of the remaining  variables (IRS phase shifts) is severely reduced under  a large number of QoS constraints coupled together, thus rendering this approach easily to get stuck at undesired suboptimal   solutions.

\begin{figure}[!t]
\centering
\includegraphics[width=0.65\textwidth]{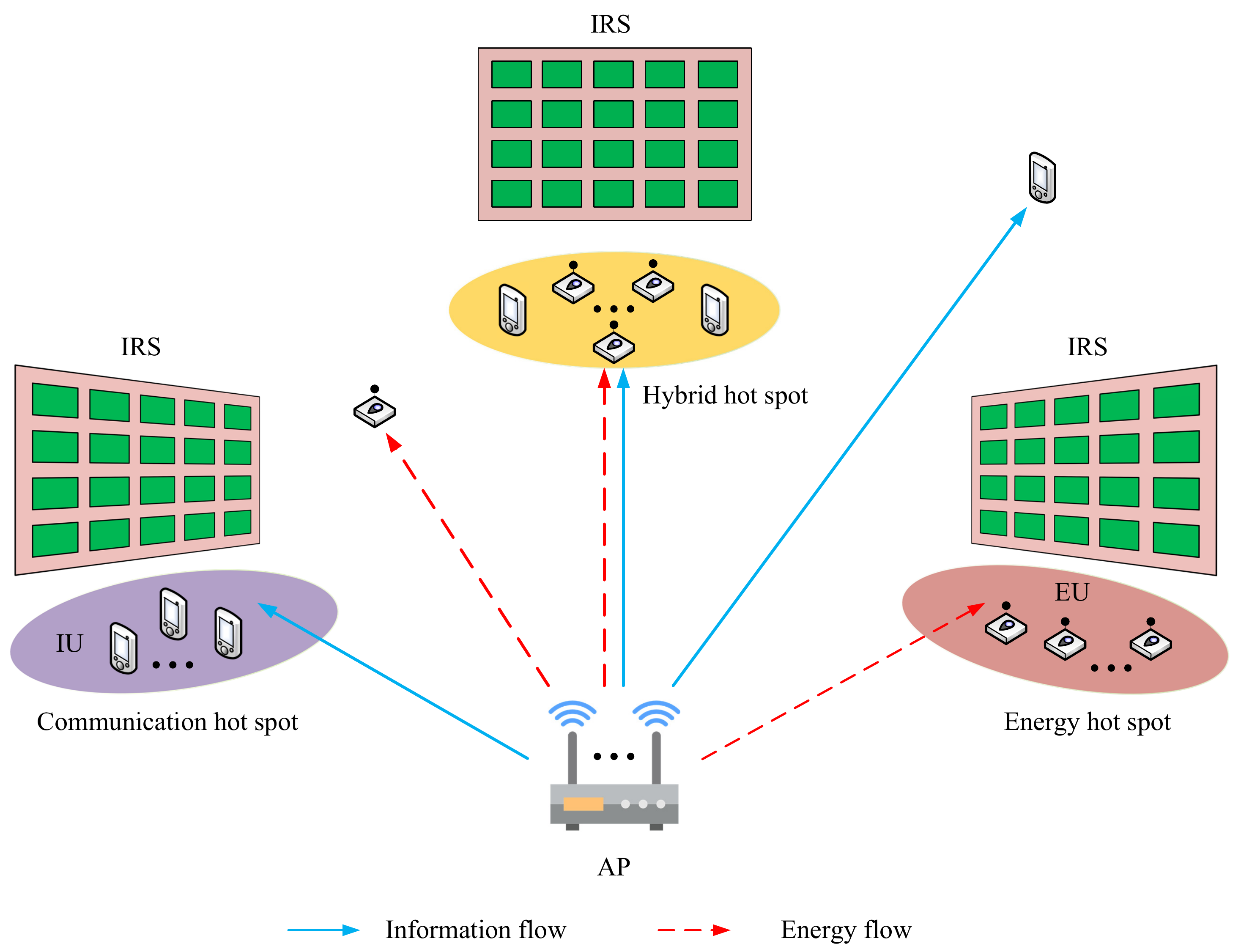}\vspace{-0.25cm} 
\caption{A SWIPT system assisted by multiple IRSs. } \label{system:model}\vspace{-0.6cm}
\end{figure}

Motivated by the above,  we study in this paper a SWIPT system assisted by multiple IRSs as shown in Fig. \ref{system:model}. Specifically, IRSs are  deployed in the hot-spot areas with high density of IUs and/or EUs to enhance their communication rate/harvested energy and thereby reduce the transmit power consumption at the access point (AP). As compared to  \cite{wu2019weighted} and \cite{pan2019intelligent} which considered the weighted sum-power maximization of EUs or weighted sum-rate maximization of IUs,  we study the QoS-constrained joint active and passive  beamforming design  in this paper, which is more challenging to solve and has not been addressed yet in the literature  to our best knowledge. Specifically, we aim to minimize the total transmit power required at the AP subject to the individual SINR constraints at IUs and energy harvesting constraints at EUs. Instead of applying the alternating optimization as in \cite{wu2019weighted} and \cite{pan2019intelligent}, we propose a novel approach to solve the new QoS-constrained joint active and passive beamforming optimization problem. To be specific, we first apply proper transformations to decouple the QoS constraints and then show that the resultant problem can be efficiently solved by jointly applying   the penalty-based and  block coordinate descent methods.  To further reduce the computational complexity, we propose to separate the phase shifts and transmit precoders optimization. As a result, we devise an alternative low-complexity algorithm that admits parallel  passive beamforming optimization for all IRSs, by exploiting their local coverage.  Numerical results show  that significant transmit power saving can be achieved at the AP   by deploying IRSs while meeting the users'   QoS requirements. Moreover, the performance gains of the proposed algorithms against other heuristically designed benchmark schemes are also shown under various practical setups, and new insights on IRS-assisted SWIPT transceiver design are also drawn.

The rest of this paper is organized as follows. Section II introduces the system model and the problem formulation for the SWIPT system aided by multiple IRSs. In Sections III and IV, we propose a penalty-based iterative algorithm and a low-complexity algorithm, respectively. Section V presents numerical results to evaluate  the performance of  the proposed algorithms. Finally, we conclude the paper  in Section VI.

\emph{Notations:} Scalars are denoted by italic letters, vectors and matrices are denoted by bold-face lower-case and upper-case letters, respectively. $\mathbb{C}^{x\times y}$ denotes the space of $x\times y$ complex-valued matrices. For a complex-valued vector $\bm{x}$, $\|\bm{x}\|$ denotes its Euclidean norm and $\text{diag}(\bm{x})$ denotes a diagonal matrix with each diagonal entry being the  corresponding entry in $\bm{x}$. The distribution of a circularly symmetric complex Gaussian (CSCG) random vector with mean vector  $\bm{x}$ and covariance matrix ${\bm \Sigma}$ is denoted by  $\mathcal{CN}(\bm{x},{\bm \Sigma})$; and $\sim$ stands for ``distributed as''. For a square matrix $\Ss$, ${\rm{tr}}(\Ss)$ and $\Ss^{-1}$ denote its trace and inverse, respectively, while $\Ss\succeq \bm{0}$ means that $\Ss$ is positive semi-definite, where $\bm{0}$ is a zero matrix of proper size.  For any general matrix $\A$, $\A^H$,  ${\rm{rank}}(\A)$, and $\A(i,j)$ denote its conjugate transpose, rank, and $(i,j)$th entry, respectively. $\I_M$  denotes an identity matrix  of size $M \times M$. $\mathbb{E}(\cdot)$ denotes the statistical expectation. $ \mathrm{Re}\{\cdot\}$ denotes the real part of a complex number. 

\section{System Model and Problem Formulation}
\subsection{System Model}
\vspace{-0.1cm}
As shown in Fig. \ref{system:model}, we consider an IRS-assisted wireless system where multiple IRSs  are deployed to assist in the SWIPT from the AP with $M$ antennas to two sets of single-antenna users, i.e.,  IUs  and EUs, denoted by $\K_{\II}=\{1, \cdots,K_{I}\}$ and $\K_{\E}=\{1, \cdots,K_{E}\}$, respectively. We assume that there are $L$ IRSs in total, indexed by $1, \cdots, L$, with  the $\ell$th IRS consisting of $N_\ell$ reflecting elements (or equivalently subgroups of adjacent elements \cite{JR:yang2019irs}).  Thus the total number of reflecting elements is given by $N= \sum^{L}_{\ell=1}N_{\ell}$ with the set of all elements denoted by $\mathcal{N}=\{1, \cdots, N\}$.    In practice, each IRS is usually attached with a smart controller that controls the phase shifts of its reflecting elements in real time and also communicates with the AP via a separate wireless link  for coordinating transmission and exchanging information on e.g. channel knowledge \cite{JR:wu2019IRSmaga}.
 For simplicity, we consider linear transmit precoding at the AP and assume that each IU/EU is assigned with one individual information/energy beam without loss of generality. Thus, the transmitted signal from the AP can be expressed as
\begin{align}
\mv{x} = \sum_{i\in {\mathcal{K_I}}}{\mv w}_i s_i^{\rm{I}} +\sum_{j\in {\mathcal{K_E}}}{\mv v}_j s_j^{\rm{E}}, \label{equa:jnl:1}
\end{align}
where ${\mv w}_i\in {\mathbb C}^{M\times 1}$ and ${\mv v}_j\in {\mathbb C}^{M\times 1}$ are the precoding vectors for IU $i$ and EU $j$, while $s_i^{\rm{I}}$ and $s_j^{\rm{E}}$ denote the information-bearing and energy-carrying signals, respectively. For information signals $s_i^{\rm{I}}$'s, they are assumed to be independent and identically distributed (i.i.d.) CSCG random variables with zero mean and unit variance, i.e., $s_i^{\rm{I}} \sim \mathcal{CN}(0,1), \forall i\in \mathcal{K_{I}}$. In contrast, since  energy signals $s_j^{\rm{E}}$'s do not carry any information, they can be any arbitrary random signals provided that their power spectral densities satisfy certain microwave radiation regulations \cite{zeng2017communications}. Without loss of generality,  we assume that  $s_j^{\rm{E}}$'s are  independently generated from an arbitrary distribution with  $\mathbb{E}\left(|s_j^{\rm{E}}|^2\right)=1,\forall j\in \mathcal{K_{E}}$. As a result, the total transmit power required at the AP is given by
\begin{align}
\mathbb{E}(\mv{x}^H\mv{x}) = \sum_{i\in {\mathcal{K_I}}}\|{\mv w}_i \|^2 +\sum_{j\in {\mathcal{K_E}}}\|{\mv v}_j\|^2.
\end{align}

Since IRS elements have no transmit RF chains, we consider a time-division duplexing (TDD) protocol for uplink and downlink transmissions and assume channel reciprocity for the channel state information (CSI) acquisition in the downlink based on the uplink training.  
To characterize the maximum performance gain brought by IRS in this paper, we assume that the CSI of all channels involved is perfectly known at the AP for the algorithm design  in each
channel coherence time, based on the various channel acquisition methods as discussed in \cite{JR:wu2019IRSmaga}.\footnote{ In general, there are two main approaches for the IRS-involved channel acquisition, depending on whether the IRS elements are equipped with receive RF chains or not \cite{JR:wu2019IRSmaga}.  For the first approach with receive RF chains, conventional channel estimation methods can be applied  for the IRS to estimate the channels of the AP-IRS and IRS-user links, respectively. In contrast, for the second approach without receive RF chains at the IRS, the IRS reflection patterns can be designed together with the uplink pilots to estimate the concatenated AP-IRS-user channels  \cite{JR:yang2019irs,JR:beixiong2019irs}. The proposed beamforming designs in this paper are applicable with both the above channel estimation methods.}
In addition, the quasi-static flat-fading  model is assumed for all channels, while the extension to the more general frequency-selective fading channels is left for our future work. Denote by $\bm{h}^H_{d,i}\in \mathbb{C}^{1\times M}$ and  $\bm{h}^H_{r,i}(\ell)\in \mathbb{C}^{1\times N_\ell}$  the baseband equivalent channels from the AP to IU $i$ and the $\ell$th IRS to IU $i$, respectively. Their counterpart channels for EU $j$ are denoted by $\bm{g}^H_{d,j}$ and  $\bm{g}^H_{r,j}(\ell)$, respectively, and the channel  from the AP to the $\ell$th IRS is denoted by $\bm{F}(\ell)\in \mathbb{C}^{N_{\ell}\times M}$.  Let  $\ttheta(\ell)  = \text{diag} (\beta_1 e^{j\theta_1}, \cdots, \beta_{N_{\ell}} e^{j\theta_{N_{\ell}} })$ denote the reflection-coefficient matrix at the  $\ell$th IRS, where  $\beta_n \in [0, 1]$ and $\theta_n\in [0, 2\pi)$ are the reflection amplitude and  phase shift of the $n$th element, respectively \cite{JR:wu2019IRSmaga}. Since  it is costly  to implement independent control of the reflection amplitude and phase shift in practice, each element is practically favorable to be  designed to maximize the  signal reflection for simplicity \cite{JR:wu2019IRSmaga,nayeri2018reflectarray,kaina2014shaping}. As such, we assume that ideally  $\beta_n=1$, $\forall n\in \N$, in the sequel of this paper.  The  signal received at IU $i$  from both the AP-user and AP-IRS-user channels can be  expressed as 
\begin{align}\label{Sec:II:signal}
{y}_i^{\rm{I}} =\left( \sum^{L}_{\ell=1}\bm{h}^H_{r,i}(\ell)\ttheta(\ell)\bm{F}(\ell)+\bm{h}^H_{d,i}  \right)\mv{x} + z_i,  i\in \mathcal{K_{I}},
\end{align}
where $z_i\sim \mathcal{CN}(0,\sigma_i^2)$ is the i.i.d. Gaussian noise at the receiver of IU $i$. By using a compact form, \eqref{Sec:II:signal} can be rewritten as
\begin{align}\label{Sec:II:signal2}
{y}_i^{\rm{I}} =(\bm{h}^H_{r,i}\ttheta\bm{F}+\bm{h}^H_{d,i})\mv{x} + z_i,  i\in \mathcal{K_{I}},
\end{align}
where $\bm{h}^H_{r,i}\in \mathbb{C}^{1\times N}$, $\ttheta  \in \mathbb{C}^{N\times N} $, and $\bm{F}\in \mathbb{C}^{N\times M}$ are respectively given by
\begin{equation}\label{Sec:II:signal3}
\bm{h}_{r,i} =\begin{bmatrix}
\bm{h}_{r,i}(1)  \\
\cdots \\
\bm{h}_{r,i}(L) \\
\end{bmatrix},~~
\ttheta =\begin{bmatrix}
\ttheta(1)& {\bm 0}  &{\bm 0} \\
 {\bm 0}  &\cdots & {\bm 0} \\
{\bm 0} & {\bm 0}  & \ttheta(L) \\
\end{bmatrix},~~
\F=\begin{bmatrix}
\F(1)  \\
\cdots \\
\F(L)\\
\end{bmatrix}.
\end{equation}
{As shown in \eqref{Sec:II:signal} and \eqref{Sec:II:signal2},  the effective channels of IUs (similarly for EUs) in the case of multiple IRSs distributed in different locations can be equivalently expressed as that in the case with a single (larger-size)   IRS.  However, if the multiple  IRSs are well separated in practice, $\bm{h}_{r,i}$ would be a sparse channel vector since each IRS has only limited signal coverage due to its passive reflection. This property  will be exploited in Section \ref{IV} to propose a low-complexity algorithm.}

Since energy beams carry no information but instead pseudorandom signals whose waveforms can be assumed to be known at both the AP and each IU before data transmission, we assume that their caused  interference can be cancelled at each IU, similarly as in \cite{xu2014multiuser}.  This facilitates us in characterizing the fundamental performance limit of SWIPT systems as well as  studying the effect of IRS on the energy beamforming.
Thus,  the SINR of IU $i$ is given by
\begin{align}\label{eq:SINR:type1}
\text{SINR}_i = \frac{|\bm{h}^H_{i} \bm{w}_i |^2}{\sum\limits_{ k\neq i, k\in \K_{\II} }| \bm{h}^H_{i}  \bm{w}_k |^2 + \sigma^2_i},   i\in \mathcal{K_{I}},
\end{align}
where $ \bm{h}^H_{i} =\bm{h}^H_{r,i}\ttheta\bm{F}+\bm{h}^H_{d,i}$.
On the other hand, by ignoring the noise power, the received RF power at EU $j$, denoted by $Q_j$, is given by
\begin{align}\label{EH:energy}
Q_j&= \sum\limits_{i\in\mathcal{K_I}}|\g^H_j {\mv w}_i|^2  + \sum \limits_{m\in\mathcal{K_E}}|\g^H_j {\mv v}_m|^2, j\in \mathcal{K_{E}},
\end{align}
where $\g^H_j = \bm{g}^H_{r,j}\ttheta \bm{F} +  \bm{g}^H_{d,j}$.
\vspace{-0.3cm}
\subsection{Problem Formulation}
Let $\bm{\theta}= [\theta_1, \cdots, \theta_N]$.  In this paper, we aim to minimize the total transmit power required at the AP subject to the individual SINR constraints at IUs and energy harvesting constraints at EUs via joint optimization of the transmit precoders at the AP and  phase shifts at all IRSs. Accordingly, the optimization problem is  formulated as
\begin{align}
\text{(P1)}: \min_{\{\bm{w}_i\}, \{\bm{v}_j\}, \bm{\theta}} & \sum_{i\in \K_{\II}}\|\bm{w}_i\|^2 + \sum_{j\in \K_{\E}}\|\bm{v}_j\|^2  \label{eq:obj}\\
\mathrm{s.t.}
~~~  &\frac{|(\bm{h}^H_{r,i}\ttheta\bm{F}+\bm{h}^H_{d,i}) \bm{w}_i |^2}{\sum\limits_{ k\neq i, k\in \K_{\II} }| (\bm{h}^H_{r,i}\ttheta\bm{F}+\bm{h}^H_{d,i} )  \bm{w}_k |^2 + \sigma^2_i} \geq \gamma_i, \forall i \in \K_{\II}, \label{P1:SINRconstrn}\\
&\sum\limits_{i\in\mathcal{K_I}}| ( \bm{g}^H_{r,j}\ttheta \bm{F} +  \bm{g}^H_{d,j}) {\mv w}_i|^2  + \sum \limits_{m\in\mathcal{K_E}}| (  \bm{g}^H_{r,j}\ttheta \bm{F} +  \bm{g}^H_{d,j}  ){\mv v}_m|^2 \geq {E}_{j}, \ \forall j\in \mathcal{K_{E}}, \label{P1:EHconstrn}\\
& 0\leq \theta_n \leq 2\pi, \forall n\in \mathcal{N}, \label{P1:phase:constraints}
\end{align}
where $\gamma_i > 0$ and ${E}_{j}>0$ are the minimum SINR and RF receive power requirements of IU $i$ and EU $j$, respectively. From \eqref{P1:EHconstrn}, one can observe that with the presence of  information beams $\bm{w}_i$'s,  dedicated energy beams $\bm{v}_i$'s may not be needed since it  is possible to meet the energy harvesting constraints  by jointly  optimizing the IRSs' phase shifts and information beams only, i.e., $\sum\limits_{i\in\mathcal{K_I}}| ( \bm{g}^H_{r,j}\ttheta \bm{F} +  \bm{g}^H_{d,j}) {\mv w}_i|^2 \geq {E}_{j}$, thus potentially  simplifying the transmitter (energy beamforming) and receiver (energy signal cancellation) designs as compared to the conventional MIMO SWIPT system without IRSs \cite{xu2014multiuser}, as will be shown later in Section V.   Note that (P1) is a non-convex optimization problem in general with the transmit precoders and IRSs' phase shifts intricately coupled in the QoS constraints.
 Generally, there is no standard method for solving such non-convex optimization problems optimally. One commonly used method  is to apply the alternating optimization to  (P1)  by iteratively  optimizing each of the transmit precoders and phase shifts with the other being fixed, as in  \cite{cui2019secure,guan2019intelligent,chen2019intelligent,dongfang2019,fu2019intelligent,huangachievable,yang2019intelligent,huang2018energy,wu2019weighted,pan2019intelligent}.  However, such alternating optimization method becomes inefficient  as the number of QoS constraints increases in (P1) since it is prone to getting  trapped at undesired suboptimal   solutions due to the more stringent coupling among  the variables  (as will be shown later in Section V). As such, in this paper, we propose a new penalty-based algorithm  to solve (P1) and  show that by applying proper reformulations to the QoS constraints in (P1),  this problem  can be efficiently solved with high-quality suboptimal  solutions.

\vspace{-0.2cm}
\section{Penalty-based Algorithm}\label{centralized:alg}
In this section, we propose a two-layer  penalty-based algorithm to solve (P1).  Specifically, the inner layer solves a penalized optimization problem  by applying the  block coordinate descent method while the outer layer updates the penalty coefficient, until the convergence is achieved.
\vspace{-0.2cm}
\subsection{Problem Reformulation for Decoupling QoS Constraints}
\vspace{-0.0cm}
The main difficulty for solving  (P1) lies in  the  QoS constraints that are  coupled in \eqref{P1:SINRconstrn} and  \eqref{P1:EHconstrn}. The key to tackle them is by  introducing new auxiliary variables to decouple them,  based on which (P1) can be efficiently solved by solving a series of simplified subproblems only.

 To this end,  let ${\bm{h}}^H_i\bm{w}_k=x_{i,k}$,   ${\bm{g}}^H_j\bm{w}_i=s_{j,i}$, and   ${\bm{g}}^H_j\bm{v}_m=t_{j,m}$, $i, k\in \K_{\II}$,  $j, m\in \K_{\E}$. Then the SINR and energy harvesting constraints can be respectively expressed as
\begin{align}
&\frac{|x_{i,i}|^2}{\sum\limits_{ k\neq i, k\in \K_{\II} }|x_{i,k} |^2 +  \sigma^2_i}
 \geq \gamma_i, \forall i \in \K_{\II},  \label{P1:new:SINRconstrn} \\
&\sum\limits_{i\in\mathcal{K_I}}|s_{j,i}|^2  + \sum \limits_{m\in\mathcal{K_E}}|t_{j,m}|^2 \geq {E}_{j}, \ \forall j\in \mathcal{K_{E}}.  \label{P1:new:EHconstrn}
\end{align}
By replacing \eqref{P1:SINRconstrn} and \eqref{P1:EHconstrn} with   \eqref{P1:new:SINRconstrn} and \eqref{P1:new:EHconstrn}, (P1) is equivalently transformed to
\begin{align}
\text{(P2)}: ~\min_{\{\bm{w}_i\}, \{\bm{v}_j\}, \bm{\theta}}& ~ \sum_{i\in \K_{\II}}\|\bm{w}_i\|^2 + \sum_{j\in \K_{\E}}\|\bm{v}_j\|^2  \label{eq:obj}\\
\mathrm{s.t.}~~~
& \frac{|x_{i,i}|^2}{\sum\limits_{ k\neq i, k\in \K_{\II} }|x_{i,k} |^2 +  \sigma^2_i}
 \geq \gamma_i, \forall i \in \K_{\II},        \label{P2:new:SINRconstrn}      \\
& \sum\limits_{i\in\mathcal{K_I}}|s_{j,i}|^2  + \sum \limits_{m\in\mathcal{K_E}}|t_{j,m}|^2 \geq {E}_{j}, \ \forall j\in \mathcal{K_{E}},  \label{P2:new:EHconstrn} \\
 & {\bm{h}}^H_i\bm{w}_k=x_{i,k},  i, k\in \K_{\II},  \label{P2:equality1}  \\
 &{\bm{g}}^H_j\bm{w}_i=s_{j,i}, {\bm{g}}^H_j\bm{v}_m=t_{j,m}, i\in \K_{\II}, j, m\in  \K_{\E},    \label{P2:equality2}  \\
& 0\leq \theta_n \leq 2\pi, \forall n\in \mathcal{N}. \label{phase:constraints}
\end{align}
Although (P2) is still a non-convex optimization problem, the optimization variables  in constraints \eqref{P2:new:SINRconstrn} and \eqref{P2:new:EHconstrn} are fully decoupled and exclusive for different IUs and EUs, i.e., no two constraints involve a common variable. More importantly, such a transformation facilitates in updating  the optimization variables in parallel, as will be detailed in the next subsection. Note that the coupling on the transmit precoders and phase shifts in (P2) are still preserved by the newly added equality constraints in  \eqref{P2:equality1} and  \eqref{P2:equality2}. To overcome them,  we exploit the gist  of the penalty-based  methods in  \cite{bertsekas1999nonlinear,Boyd,shi2016joint} by integrating such constraints  into the objective function of (P2). Specifically, we convert the equality constraints in  \eqref{P2:equality1} and  \eqref{P2:equality2} into quadratic functions and then add them as a penalty term in the objective function of (P2), yielding the following optimization problem
\begin{align}
\text{(P3)}:~\min_{\{\bm{w}_i\}, \{\bm{v}_j\}, \bm{\theta}, \{x_{i,k},s_{j,i},t_{j,m}\} }  &  \sum_{i\in \K_{\II}}\|\bm{w}_i\|^2  + \sum_{j\in \K_{\E}}\|\bm{v}_j\|^2 +\frac{1}{2\rho}\left(  \sum_{i\in  \K_{\II} }  \sum_{k\in  \K_{\II}}  |{\bm{h}}^H_i\bm{w}_k-x_{i,k}|^2  \right.  \nonumber \\
                 & \left.  + \sum_{j\in  \K_{\E} }  \sum_{i\in  \K_{\II}}   |{\bm{g}}^H_j\bm{w}_i-s_{j,i}|^2 + \sum_{j\in  \K_{\E} }  \sum_{m\in  \K_{\E}}  |{\bm{g}}^H_j\bm{v}_m-t_{j,m}|^2 \right)  \label{P3:obj}\\
\mathrm{s.t.}~~~~~~~~~~
& \frac{|x_{i,i}|^2}{\sum\limits_{ k\neq i, k\in \K_{\II} }|x_{i,k} |^2 +  \sigma^2_i}
 \geq \gamma_i, \forall i \in \K_{\II},     \label{P3:new:SINRconstrn}    \\
&  \sum\limits_{i\in\mathcal{K_I}}|s_{j,i}|^2  + \sum \limits_{m\in\mathcal{K_E}}|t_{j,m}|^2 \geq {E}_{j}, \ \forall j\in \mathcal{K_{E}},     \label{P3:new:EHconstrn}    \\
& 0\leq \theta_n \leq 2\pi, \forall n\in \mathcal{N}, \label{P3:new:phase:constraints}
\end{align}
where $\rho>0$ denotes the penalty coefficient used for penalizing  the violation of equality constraints in (P2).
It is worth pointing out that although the equality constraints are relaxed in (P3), when $\rho \rightarrow 0$ $(1/\rho \rightarrow \infty)$, the solution obtained by solving  (P3) always satisfies all equality constraints in (P2). However, it is practically undesirable to initialize $\rho$ to be a very small value, since in this case the penalized objective function in (P3) will be dominated by the quadratic penalty terms and thus the original objective function (i.e., the transmit power in (P2)) will be diminished, rendering this approach ineffective. In contrast, initializing $\rho$ to be a sufficiently large value helps obtain a good starting point for the proposed algorithm, even though this point may be infeasible for (P2). By gradually decreasing the value of $\rho$, we can minimize the transmit power and also obtain a solution that satisfies all the equality constraints within a predefined accuracy.

For any given $\rho>0$, (P3) is still a non-convex optimization problem due to the non-convex objective function as well as non-convex constraints in \eqref{P3:new:SINRconstrn} and \eqref{P3:new:EHconstrn}. However, it is observed  that  each optimization variable in (P3) is involved in at most one constraint, which thus motivates us to apply the block coordinate descent method to solve  (P3) efficiently by  properly partitioning the optimization variables into different blocks.
  Specifically, the entire optimization variables can be   partitioned into three blocks, i.e., 1) transmit precoders at the AP, i.e., $\{\bm{w}_i\}$ and $\{\bm{v}_j\}$, $i \in \K_{\II}$,  $j\in \K_{\E}$, 2) phase shifts at IRSs, i.e.,  $\bm{\theta}$, and 3) auxiliary variables, i.e.,   $\{x_{i,k}, s_{j,i},t_{j,m}\}$, $i, k\in \K_{\II}$,  $j, m\in \K_{\E}$. Then,    we can  minimize the penalized objective function in (P3) by alternately  optimizing each of the above three blocks  in one iteration with the other two blocks  fixed, and iterating the above until the convergence is reached. The details are provided  in the next subsection  and  the convergence is achieved in the inner layer until the fractional decrease of the objective function of (P3) is less than a sufficiently small threshold $\epsilon_1>0$ or a maximum number of iterations is reached.



\subsection{Inner Layer: Block Coordinate Descent Algorithm for Solving  (P3)}
1) For any given phase shifts $\bm{\theta}$ and auxiliary variables $\{x_{i,k}, s_{j,i},t_{j,m}\}$, $i, k\in \K_{\II}$,  $j, m\in \K_{\E}$, the transmit precoders in (P3) can be optimized by solving the following problem
\begin{align}
&\text{(P3.1)}:~\min_{\{\bm{w}_i\}, \{\bm{v}_j\} }    \sum_{i\in \K_{\II}}\|\bm{w}_i\|^2  + \sum_{j\in \K_{\E}}\|\bm{v}_j\|^2 +\frac{1}{2\rho}\left(  \sum_{i\in  \K_{\II} }  \sum_{k\in  \K_{\II}}  |{\bm{h}}^H_i\bm{w}_k-x_{i,k}|^2  \right.  \nonumber  \\
                 &~~~~~~~~~~~~~~~~~~~~ \left.  +  \sum_{j\in  \K_{\E} }  \sum_{i\in  \K_{\II}}   |{\bm{g}}^H_j\bm{w}_i-s_{j,i}|^2 + \sum_{j\in  \K_{\E} }  \sum_{m\in  \K_{\E}}  |{\bm{g}}^H_j\bm{v}_m-t_{j,m}|^2 \right)  \label{eq:obj}
\end{align}
It is not difficult to observe that (P3.1) is a convex quadratic minimization problem without constraint, for which the optimal solution can be readily obtained by exploiting the first-order optimality condition of the objective  function \cite{Boyd}.  By setting the first-order derivatives of the objective function with respect to  $\bm{w}_i$ and $\bm{v}_j$ equal to zero, respectively, the optimal transmit precoders to (P3.1) are obtained in closed-form expressions given by
\begin{align}
\bm{w}^*_i &=\frac{1}{2\rho}\A_1^{-1}(    \sum_{k\in \K_{\II}}  \h_k x_{k,i}  +      \sum_{j\in \K_{\E}}  \g_j s_{j,i}  ),  \label{transmit:if:bf} \\
\bm{v}^*_j &= \frac{1}{2\rho}\A_2^{-1}(          \sum_{m\in \K_{\E}}  \g_m t_{m,j}  ), \label{transmit:energy:bf}
\end{align}
where
\begin{align}
\A_1 =  \I_M + \frac{      \sum_{k\in \K_{\II}}  \h_k \h_k^H  +      \sum_{m\in \K_{\E}}  \g_m \g_m^H }{2\rho},
\A_2 =  \I_M + \frac{        \sum_{m\in \K_{\E}}  \g_m \g_m^H }{2\rho}.
\end{align}
It is worth pointing out that all $\bm{w}^*_i$'s and $\bm{v}^*_i$'s for different IUs and EUs can be updated  in parallel by using \eqref{transmit:if:bf} and \eqref{transmit:energy:bf}.

2)  For any given  transmit precoders $\bm{w}_i$'s and $\bm{v}_i$'s and auxiliary variables  $\{x_{i,k}, s_{j,i},t_{j,m}\}$, $i, k\in \K_{\II}$,  $j, m\in \K_{\E}$, the phase shifts can be optimized by solving (P3) with constraints only  in \eqref{P3:new:phase:constraints}.  Let $\bm{u} = [u_1, \cdots, u_N]^H$ where $u_n = e^{j\theta_n}$, $\forall n$. Then, constraints in \eqref{P3:new:phase:constraints} are equivalent to  the unit-modulus constraints:  $|u_n|^2=1, \forall n$. By applying the change of variables $-\bm{h}^H_{d,i}\w_k+ x_{i,k} =\bar{C}_{i,k}$ and $\bm{h}^H_{r,i}\ttheta \bm{F}\w_k =\bm{u}^H \bar{\dd}_{i,k} $ where $\bar{\dd}_{i,k} = \text{diag}(\bm{h}^H_{r,i})\bm{F}\w_k \in \mathbb{C}^{N\times 1}$, we have
\begin{align}
 {\bm{h}}^H_i\bm{w}_k-x_{i,k} = \bm{h}^H_{r,i}\ttheta\bm{F}\w_k + \bm{h}^H_{d,i}\w_k -x_{i,k}  =  {\bm{u}}^H\bar{\dd}_{i,k} - \bar{C}_{i,k}.
\end{align}
Similarly, we have ${\bm{g}}^H_j\bm{w}_i-s_{j,i}  = {\bm{u}}^H\check{\dd}_{j,i} - \check{C}_{j,i}$ and  ${\bm{g}}^H_j\bm{v}_m-t_{j,m} = {\bm{u}}^H\hat{\dd}_{j,m} - \hat{C}_{j,m}$ for the other two quadratic penalty terms in \eqref{P3:obj}.  As a result,  the subproblem regarding to phase shifts optimization is given by (with  constant terms ignored)
\begin{align}                                
\text{(P3.2)}:~\min_{ \bm{u} }  ~~ &  \sum_{i\in  \K_{\II} }  \sum_{k\in  \K_{\II}}  |{\bm{u}}^H\bar{\dd}_{i,k} - \bar{C}_{i,k}|^2
  + \sum_{j\in  \K_{\E} }  \sum_{m\in  \K_{\E}}    |{\bm{u}}^H\hat{\dd}_{j,m} - \hat{C}_{j,m}|^2   \nonumber \\ & + \sum_{j\in  \K_{\E} }  \sum_{i\in  \K_{\II}}   |{\bm{u}}^H\check{\dd}_{j,i} - \check{C}_{j,i}|^2  \label{eq:obj} \\
~~~\mathrm{s.t.}~~ & |u_n|=1, \forall n\in\mathcal{N}. \label{P3.2:unit:modulus:constraint}
\end{align}
It is noted that (P3.2) is non-convex due to the non-convex unit-modulus constraints in \eqref{P3.2:unit:modulus:constraint}.  However, since the phase shifts of all elements are fully separable in constraints and only coupled in the objective function, we can apply the block coordinate descent method for optimizing them iteratively.  Specifically, for a given $n\in \mathcal{N}$ in (P3.2),  by fixing ${u}_{n'}$'s, $\forall n'\neq n, n'\in \mathcal{N}$, we observe that the objective function of (P3.2) is  linear with respect to ${u}_n$ and can be written as
\begin{align}\label{Step2:objective}
&2\mathrm{Re}\left\{u_n\varphi_n\right\} + \sum_{n'\neq n}^{N}\sum_{m\neq n}^{N}\R(n',m)u_{n'}u_m^H+C,
\end{align}
where
\begin{align}
\varphi_n &= \sum_{n'\neq n}^{N}\R(n,n')u_{n'}^H-  \bb(n),  \label{obj:phase2} \\
\R & =    \sum_{i\in  \K_{\II} }  \sum_{k\in  \K_{\II}} \bar{\dd}_{i,k}\bar{\dd}_{i,k}^H
   + \sum_{j\in  \K_{\E} }  \sum_{m\in  \K_{\E}} \hat{\dd}_{j,m} \hat{\dd}_{j,m}^H +  \sum_{j\in  \K_{\E} }  \sum_{i\in  \K_{\II}} \check{\dd}_{j,i}\check{\dd}_{j,i}^H, \\
\bb & =   \sum_{i\in  \K_{\II} }  \sum_{k\in  \K_{\II}} \bar{\dd}_{i,k}\bar{C}_{i,k}^H
   + \sum_{j\in  \K_{\E} }  \sum_{m\in  \K_{\E}} \hat{\dd}_{j,m}\hat{C}_{j,m}^H +  \sum_{j\in  \K_{\E} }  \sum_{i\in  \K_{\II}} \check{\dd}_{j,i} \check{C}_{j,i}^H, \label{Step2:objective:3}\\
   C  &=  \R(n,n)- 2\mathrm{Re} \Big\{\sum_{n'\neq n}^{N}u_{n'}\bb(n')\Big\}   +   \sum_{i\in  \K_{\II} }  \sum_{k\in  \K_{\II}}   |\bar{C}_{i,k}|^2 \nonumber \\
       &~~~+   \sum_{j\in  \K_{\E} }  \sum_{m\in  \K_{\E}}    |\hat{C}_{j,m}|^2 + \sum_{j\in  \K_{\E} }  \sum_{i\in  \K_{\II}} |\check{C}_{j,i}|^2.
\end{align}
Based on \eqref{Step2:objective}, it is not difficult to show that the optimal solution of ${u}_n$ to (P3.2) is given by
\begin{align}\label{optimal:trj}
{{u}}^*_n=\left\{\begin{aligned}
                  &    1,&\text{if}~ \varphi_n=0, \\
                      & \frac{\varphi_n^H}{ |\varphi_n| }, &\text{otherwise}.
                    \end{aligned}
             \right.
\end{align}
Based on \eqref{optimal:trj}, we alternately optimize each of the $N$  phase shifts in an iterative manner by fixing the other $N-1$ phase shifts, until the convergence is achieved in this block.

3)   For any given transmit precoders $\bm{w}_i$'s and $\bm{v}_j$'s and phase shifts $\bm{\theta}$, the  auxiliary variables can be optimized by solving (P3) with  constraints in \eqref{P3:new:SINRconstrn} and \eqref{P3:new:EHconstrn}. Since the optimization variables with respect to different IUs and EUs are separable in both the objective function and constraints,  we can solve the resultant problem by solving  $K_I + K_E$ independent subproblems in parallel, each with only one single (either SINR or energy harvesting) constraint.  Specifically, for IU $i$, the corresponding subproblem with respect to $x_{i,k}$'s, $\forall k\in \K_{\II}$, is reduced to (by ignoring constant terms) 
\begin{align}\label{QCQP1}
\text{(P3.3)}:~\min_{ \{x_{i,k}, \forall  k\} } &~   \sum_{k\in  \K_{\II}}  | {\bar x}_{i,k} -x_{i,k}|^2     \\
\mathrm{s.t.}
~&~~ \frac{|x_{i,i}|^2}{\sum\limits_{ k\neq i, k\in \K_{\II} }|x_{i,k} |^2 + \sigma^2_i}
 \geq \gamma_i, \label{P3.3:SINRconstrn}
\end{align}
where $ {\bar x}_{i,k} = {\bm{h}}^H_i\bm{w}_k$,  $i, k\in \K_{\II}$. Note that if ${\bar x}_{i,i}=0$,  the constraint in \eqref{P3.3:SINRconstrn} will always be met with equality at the optimal solution since otherwise $|{ x}_{i,i}|$ can be further reduced to decrease the objective value. As such, (P3.3) can be transformed to
\begin{align}\label{QCQP1:1}
\min_{ \{x_{i,k}, \forall  k\neq i, k\in \K_{\II}\} } &~   \sum_{k\neq i, k\in \K_{\II}}  | {\bar x}_{i,k} -x_{i,k}|^2    + \gamma_i \sum\limits_{ k\neq i, k\in \K_{\II} }|x_{i,k} |^2 + \gamma_i\sigma^2_i,
\end{align}
which is a  convex quadratic minimization problem without constraint as (P3.1) and thus can be similarly solved (note that all $K_I-1$ variables are also decoupled in the objective function). However, in the general case with ${\bar x}_{i,i}\neq 0$,  (P3.3) is a non-convex quadratically constrained quadratic program (QCQP) with one single constraint. Fortunately,  it has been shown in \cite{Boyd} that strong duality holds for this type of non-convex problems, provided that the Slater's constraint qualification is satisfied. As such, the duality gap between (P3.3) and its dual problem is zero, which means that the optimal solution can be obtained efficiently by applying the Lagrange duality method.  Denote by $\lambda_i\geq 0$ the dual variable associated with constraint \eqref{P3.3:SINRconstrn}. The  Lagrangian associated with (P3.3) can be expressed as (by ignoring constant terms)
\begin{align}\label{lagrange:function}
\mathcal{L}(\{x_{i,k}\}, \lambda_i) = (1-\lambda_i)|x_{i,i}|^2 + \sum_{k\neq i, k\in \K_{\II}}(1+\lambda_i\gamma_i) |x_{i,k}|^2 - 2\sum_{k\in  \K_{\II}}\mathrm{Re}\{ {\bar x}_{i,k}x_{i,k}^H \}.
\end{align}
Accordingly, the dual function is given by $f(\lambda_i)= \min_{ \{x_{i,k}, k\in \mathcal{ \K_{\II} \} }}  \mathcal{L}(\{x_{i,k}\}, \lambda_i)$, for which the following lemma holds.
\begin{lemma}\label{lem:multiplier}
To make $f( \lambda_i)$ bounded from the below, i.e., $f( \lambda_i) > - \infty$,  it follows that $\lambda_i < 1$ must hold.
\end{lemma}
\begin{proof}
This is shown by contradiction. Based on ${\bar x}_{i,i}\neq 0$, if $\lambda_i \geq  1$, we have $f( \lambda_i) \rightarrow - \infty$ by setting $|x_{i,i}| \rightarrow \infty$. Thus, this lemma is proved.
\end{proof}
Based on  Lemma \ref{lem:multiplier}, by exploiting the first-order optimality condition, the optimal solution to minimize the Lagrangian in \eqref{lagrange:function} for fixed $\lambda_i$ is given by
\begin{align}
x^{\star}_{i,i} &=  \frac{  {\bar x}_{i,i}   }{1-\lambda_i}, \label{QCQP:solution1}  \\
x^{\star}_{i,k} &=\frac{ {\bar x}_{i,k}  }{1+\lambda_i\gamma_i }, k\neq i, k\in \K_{\II}.  \label{QCQP:solution2}
\end{align}
If the SINR constraint in \eqref{P3.3:SINRconstrn} is not met with equality at the optimal solution, i.e., $\lambda_i =0$, then the optimal solution to (P3.3) is given by $ {x}^*_{i,k}= {\bar x}_{i,k}$, $i, k\in \K_{\II}$. Otherwise, by substituting \eqref{QCQP:solution1}-\eqref{QCQP:solution2} into \eqref{P3.3:SINRconstrn}, this equality constraint can be written as
\begin{align} \label{QCQP:solution4}
\mathcal{G}(\lambda_i) \triangleq  \frac{  |{\bar x}_{i,i}|^2   }{ (1 - \lambda_i)^2 } -  \sum_{ k\neq i, k\in \K_{\II} }\frac{ \gamma_i |{\bar x}_{i,k}|^2  }{ (1+\lambda_i\gamma_i)^2 }    - \gamma_i\sigma^2_i =0.
\end{align}
It is not difficult to show that $\mathcal{G}(\lambda_i)$ is a monotonically increasing function of $\lambda_i$ for $0 \leq  \lambda_i <  1$.  As such, the optimal dual variable and primal variables can be efficiently obtained by using the simple bisection search.

On the other hand, for  EU $j$, the corresponding subproblem is given by
\begin{align}\label{QCQP2}
\text{(P3.4)}:~\min_{ \{t_{j,m},s_{j,i}, \forall m,i \} } &~   \sum_{m\in  \K_{\E}}  |   {\bar t}_{j,m}   -  t_{j,m} |^2   +  \sum_{j\in  \K_{\E}}  | {\bar s}_{j,i}   -  s_{j,i} |^2  \\
\mathrm{s.t.}
&~ \sum\limits_{i\in\mathcal{K_I}}|s_{j,i}|^2  + \sum \limits_{m\in\mathcal{K_E}}|t_{j,m}|^2\geq {E}_{j}, \label{QCQP2:constraint}
\end{align}
where $ {\bar t}_{j,m} = {\bm{g}}^H_j\bm{v}_m$ and $ {\bar s}_{j,i} = {\bm{g}}^H_j\bm{w}_i $, $j, m\in \K_{\E}, i \in  \K_{\II}$.  Since (P3.4) is also a QCQP with one single constraint as (P3.3), it can be similarly solved by applying  the bisection search as proposed above and the details are thus omitted for brevity.

\subsection{Outer Layer: Update Penalty Coefficient}
Recall that the equality constraints in (P2) need to be satisfied in the converged solution of the proposed algorithm.
To this end, we gradually decrease  the value of the penalty coefficient $\rho$  as follows
\begin{align}\label{penalty:coefficient}
\rho: = c \rho,~~ 0 < c<1,
\end{align}
where $c$ is a constant scaling factor in the outer layer. Generally, a larger value of $c$ can achieve better performance but  at the cost of more iterations in the outer layer.

 \begin{algorithm}[t]
\caption{Proposed penalty-based algorithm.}\label{Alg:PDD}
\begin{algorithmic}[1]

\STATE Initialize $\bm{\theta}$,  $\{x_{i,k}, s_{j,i},t_{j,m}\}$, $i, k\in \K_{\II}$,  $j, m\in \K_{\E}$, and $\rho$.
\REPEAT
\REPEAT
\STATE Update transmit precoders by solving (P3.1).
\STATE Update phase shifts by solving (P3.2).
\STATE Update  auxiliary variables by solving (P3.3) and (P3.4), respectively.
\UNTIL{  The fractional decrease of the objective value of (P3) is below a threshold $\epsilon_1>0$ or the maximum number of inner iterations is reached.}
\STATE Update the penalty coefficient $\rho$ by using \eqref{penalty:coefficient}.
\UNTIL{ The constraint violation $\xi$  is below a threshold $\epsilon_2>0$.}

\end{algorithmic}
\end{algorithm}
\subsection{Convergence Analysis and Computational Complexity}
For any solution obtained  to (P3), to evaluate whether it violates the  equality constraints in (P2) or not, we adopt an indicator $\xi$ defined as
\begin{align}\label{constraint:violation}
\xi = \max \{   |{\bm{h}}^H_i\bm{w}_k-x_{i,k}|^2, |{\bm{g}}^H_j\bm{w}_i-s_{j,i}|^2,
  &|{\bm{g}}^H_j\bm{v}_m-t_{j,m}|^2,   i, k\in \K_{\II}, j, m\in  \K_{\E} \}.
\end{align}
The proposed algorithm is terminated when $\xi\leq \epsilon_2$ where $\epsilon_2$ is a predefined accuracy for all equality constraints. With the decrease of the penalty coefficient, the penalty term becomes larger and will eventually guarantee the equality constraints \cite{bertsekas1999nonlinear,Boyd,shi2016joint,zhao2019decoding}. In addition, for any $\rho$,   the objective value of (P3) achieved by applying the block coordinate descent method  is non-increasing over iterations in the inner layer and  the optimal objective value of (P3) is bounded from below. Thus, based on the result in \cite{shi2016joint},  the proposed algorithm is guaranteed to converge to a stationary point of (P1). The details of this algorithm are summarized in Algorithm 1.


Algorithm 1 is computationally efficient as the  optimization variables in lines 4 and 5 are updated by using closed-form expressions, and those in line 6 are obtained by using the simple bisection search. To be specific, it can be shown that the complexity of solving (P3.1) is $\mathcal{O}( (N^2+MN+M^2)(K_I+K_E) +M^3 )$, that of solving (P3.2) is $\mathcal{O}( N^2(K_I^2+K_E^2+K_IK_E) +I_0N)$ where $I_0$ is the number of iterations required for convergence, and that of solving (P3.3) is $\mathcal{O}( (K_I^2+K_E^2+K_IK_E) \log_2(1/\epsilon_3))$ where $\epsilon_3$ is the accuracy for the bisection search.   Thus, the overall complexity of Algorithm 1 can be written as $\mathcal{O}(I_{inn}I_{out}(  M^3 +  (MN+M^2)(K_I+K_E) + N^2(K_I^2+K_E^2+K_IK_E) +I_0N +(K_I^2+K_E^2+K_IK_E) \log_2(1/\epsilon_3)  ) )$ where $I_{inn}$ and $I_{out}$ denote respectively the outer and inner iteration numbers required for convergence.

\begin{remark}
\emph{In practice, it is generally desirable to implement IRS with discrete phase shifters \cite{JR:wu2019IRSmaga,JR:wu2019discreteIRS}, which means that the phase shift at each element of IRS only takes a finite number of discrete values. Specifically, by assuming that the discrete phase-shift values are obtained by uniformly quantizing the interval $[0,2\pi)$,  we can replace the constraints in \eqref{P1:phase:constraints}  by  
\begin{align}\label{P1:discrete:phase}
\theta_n \in \mathcal{F}\triangleq \{0,\Delta\theta, \cdots, (2^b-1)\Delta\theta \}, \forall n \in \mathcal{N},
\end{align}
where  $\Delta\theta= 2\pi/2^b$ and $b$ denotes the number of bits used to indicate the number of phase-shift levels at each element.
Unfortunately, this  renders the modified problem of (P1) to be a mixed-integer non-linear program (MINLP) that is  more challenging to solve in general. However,  Algorithm 1 is still applicable to solving the new problem  with only some slight modification. Specifically, we only need to replace (P3.2) by the following optimization problem
\begin{align}
\text{(P3.5)}:~~\min_{ \bm{u} }  ~~~ &  \sum_{i\in  \K_{\II} }  \sum_{k\in  \K_{\II}}  |{\bm{u}}^H\bar{\dd}_{i,k} - \bar{C}_{i,k}|^2   + \sum_{j\in  \K_{\E} }  \sum_{m\in  \K_{\E}}    |{\bm{u}}^H\hat{\dd}_{j,m} - \hat{C}_{j,m}|^2    \nonumber\\ & + \sum_{j\in  \K_{\E} }  \sum_{i\in  \K_{\II}}   |{\bm{u}}^H\check{\dd}_{j,i} - \check{C}_{j,i}|^2  \label{eq:obj} \\
~~~\mathrm{s.t.}~~~ &u_n \in \mathcal{F}, \forall n \in \mathcal{N}. \label{P3.5:phase:constraints}
\end{align}
As (P3.5) has a similar structure  as (P3.2) except that the constraint \eqref{P3.2:unit:modulus:constraint} is now replaced by \eqref{P3.5:phase:constraints},  it can be  solved by using the proposed method in Section III-B similarly. As a result, the optimal solution given in (36) is modified as $u^{**}_n =  \arg \min_{u_n \in \mathcal{F}  } | u_n -u^*_n|, \forall n$.}
\end{remark}

\section{Alternative Low-Complexity  Algorithm}\label{IV}
Although Algorithm 1 proposed in the preceding  section yields high-quality converged solution for (P1), the phase shifts of all IRSs' elements need to be successively optimized in each of the inner layer iterations with given $\rho$, which accounts for the main complexity of Algorithm 1.  To overcome this issue, we propose an alternative  low-complexity algorithm in this section by separating the design of the phase shifts and transmit precoders. The key idea is to exploit the short-range/local coverage of  IRSs, i.e.,  the reflected signals received at each user are mainly from its nearest  IRS in its proximity (if any), since IRSs are usually deployed to be sufficiently far apart from each other in practice to avoid complicated inter-IRS interference management. Motivated by this, for each IRS, we can first optimize its phase shifts regardless of those of the other IRSs by only considering the users associated with it,  and then optimize the transmit precoders of all users to guarantee their  QoS requirements, elaborated as follows.

First, we  associate each of the users  with an  IRS that is closest to it.  Denote the set of users (including IUs and/or EUs) associated with the $\ell$th IRS  by $\mathcal{U}_\ell$.
Use $\bm{q}^H_{d,k}$ to denote either  $\bm{h}^H_{d,k}$ or $\bm{g}^H_{d,k}$ and let  $\bm{\Phi}_k(\ell) =  \text{diag}(\bm{q}^H_{r,k}(\ell))\bm{F}(\ell)  \in \mathbb{C}^{N_{\ell} \times M}$, $ k \in \mathcal{U}_\ell$ where $\bm{q}^H_{r,k}(\ell)$ can be either $\bm{h}^H_{r,k}(\ell)$ or $\bm{g}^H_{r,k}(\ell)$.
Then the effective channel sum-power gain of all users associated with the $\ell$th IRS can be expressed as
\begin{align}\label{eq:obj}
\sum_{k\in \mathcal{U}_\ell}  \| \bm{q}^H_{r,k}(\ell)\ttheta(\ell)\bm{F}(\ell)+\bm{q}^H_{d,k}  \|^2=\sum_{k\in \mathcal{U}_\ell}   \|\uuu^H \bm{\Phi}_k(\ell) + \bm{q}^H_{d,k}  \|^2,
  \end{align}
  where $\bm{u}  = [u_1, \cdots, u_{N_\ell}]^H$ with $u_n = e^{j\theta_n}$,  $n=1,\cdots, N_{\ell}$.
   Based on \eqref{eq:obj}, the phase shifts of the $\ell$th IRS can be  optimized  by solving the following problem
\begin{align}
\text{(P4)}: ~~\max_{\uuu} ~~~&\sum_{k\in \mathcal{U}_\ell}  \|\uuu^H \bm{\Phi}_k(\ell) + \bm{q}^H_{d,k}  \|^2 \\
\mathrm{s.t.}~~~~&  |u_n|=1, n=1,\cdots, N_{\ell}. \label{eq:modulus2}
\end{align}
Since (P4) has a similar form as (P3.2), it can be efficiently solved by using the proposed method in Section III-B. More importantly, note that  the phase shifts of different IRSs can be optimized in parallel by solving each corresponding problem of (P4). With the obtained phase shifts, the effective channels of all users can be constructed (see,  e.g., \eqref{Sec:II:signal3} for IUs).
Then, the transmit precoders can be optimized at the AP by solving the following problem
\begin{align}
\text{(P5)}: \min_{\{\bm{w}_i\}, \{\bm{v}_j\} } & \sum_{i\in \K_{\II}}\|\bm{w}_i\|^2 + \sum_{j\in \K_{\E}}\|\bm{v}_j\|^2  \label{P5:eq:obj}\\
\mathrm{s.t.}
~~~  &\frac{|\bm{h}^H_{i} \bm{w}_i |^2}{\sum\limits_{ k\neq i, k\in \K_{\II} }| \bm{h}^H_{i}  \bm{w}_k |^2   + \sigma^2_i} \geq \gamma_i, \forall i \in \K_{\II}, \label{P5:SINRconstrn}\\
&\sum\limits_{i\in\mathcal{K_I}}| \bm{g}^H_{j}{\mv w}_i|^2  + \sum \limits_{m\in\mathcal{K_E}}| \bm{g}^H_{j}{\mv v}_m|^2 \geq {E}_{j}, \ \forall j\in \mathcal{K_{E}}. \label{P5:EHconstrn}
\end{align}
Since (P5) is a special case of (P1) without phase shifts,  it can be solved by Algorithm 1 in Section III.  Based on (P4) and (P5), the overall complexity of the above algorithm (refereed to as Algorithm 2) can be shown to be  $\mathcal{O}( I_0N+ I_{inn}I_{out}(M^3 +  M^2(K_I+K_E)  +(K_I^2+K_E^2+K_IK_E) \log_2(1/\epsilon_3) )  )$, where the complexity of optimizing the phase shifts (i.e., regarding $N$) is significantly reduced  as compared to that of Algorithm 1, due to this separate design.

\section{Simulation Results}\label{simulation:sec}

 \begin{figure}[ht]
\centering
\includegraphics[width=0.99\textwidth]{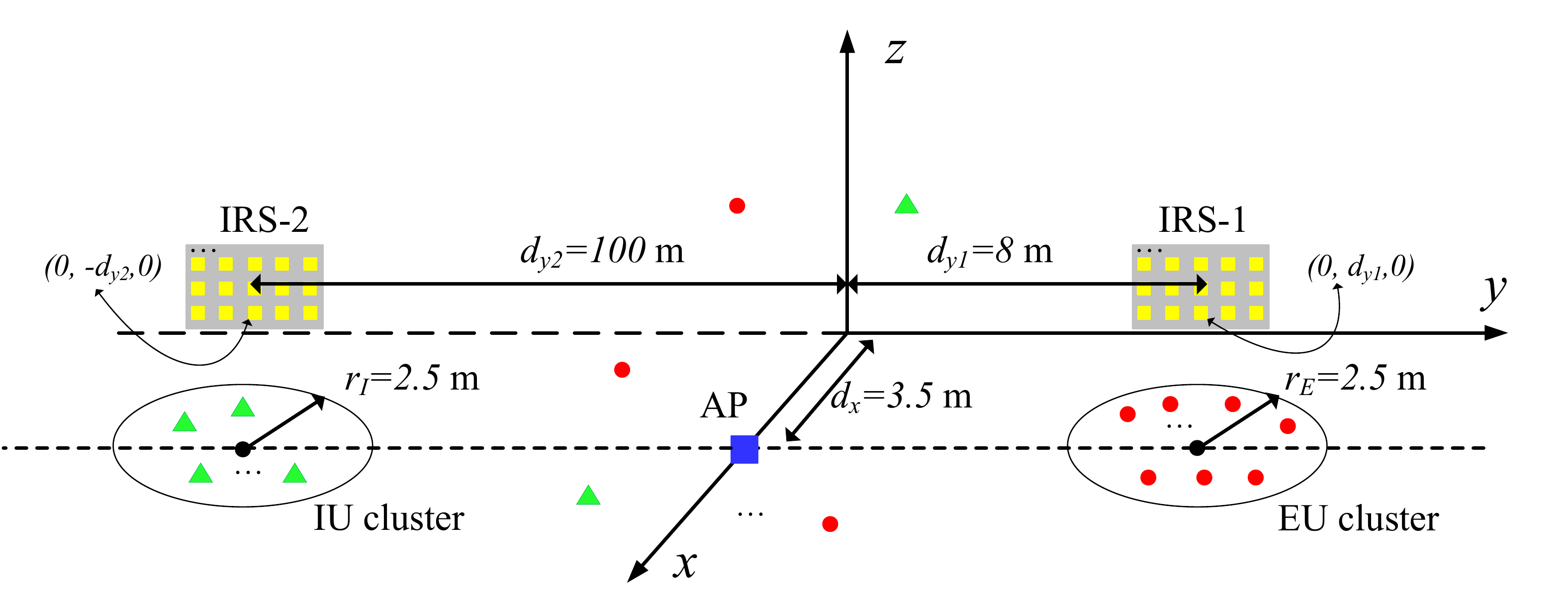}\vspace{-4mm}
\caption{Simulation setup. } \label{simulation:setup:multiIRS:SWIPT}\vspace{-8mm}
\end{figure}

In this section, numerical examples are provided to evaluate  the performance  of the proposed algorithms.    We consider a system that operates on a carrier frequency of 750 MHz with the system bandwidth of $1$ MHz and the effective noise power density of $-150$ dBm/Hz. A three-dimensional (3D) coordinate setup is considered as shown in Fig. \ref{simulation:setup:multiIRS:SWIPT}, where the AP is located  in $x$-axis with the coordinate denoted by $(d_x ,0,0)$. Besides,  two IRSs, namely IRS-1 and IRS-2,  are deployed in $y$-$z$ plane to establish hot spots for two clusters of EUs and IUs, respectively, where the cluster circles are centered at  $(d_x ,d_{y1},0)$ and   $(d_x ,-d_{y2},0)$ with radius $r_{E}$ and $r_{I}$, respectively. %
The reference elements of IRS-1 and IRS-2 are respectively located at $(0,d_{y1},0)$ and $(0,-d_{y2} ,0)$, both with a spacing of  half-wavelength, i.e., $\lambda/2 = 0.2$ m,  among adjacent elements. We assume that both IRSs have the same number of reflecting elements, i.e.,  $N=2N_0$, and for each IRS, we set $N_{0}=N_{y}N_{z}$ where $N_{y}$ and $N_{z}$ denote the numbers of reflecting elements along  $y$-axis and  $z$-axis, respectively.   For the purpose of exposition, we fix $N_y=5$ and increase $N_{z}$ linearly with $N_{0}$.  The distance-dependent path loss model is given by
\begin{align}\label{pathloss}
L(d) = C_0\left( \frac{d}{D_0} \right)^{-\alpha},
\end{align}
where $C_0= (\lambda/4\pi)^2$ is the path loss at the reference distance $D_0=1$ meter (m), $d$ denotes the individual link distance, and $\alpha$ denotes the path loss exponent.

We adopt the plane-wave model for both the AP-IRS and AP-user links, whereas the spherical-wave model for IRS-user link due to the generally limited signal coverage of  IRS, which means that the distance between each reflecting element and one user is calculated separately  based on their 3D coordinates.  Each antenna at the AP is assumed to have an isotropic radiation pattern with 0 dBi antenna gain, while each reflecting element of  IRSs is assumed to have 3 dBi gain for fair comparison, since each IRS reflects signals only in its front half-space. The path loss exponents of the AP-IRS and IRS-user links are set to be 2.2 whereas that of the AP-user link is set to be 3.8 as IRSs are usually deployed for users with weak AP-user channels and their locations can be properly selected to avoid severe blockage with the AP.
To account for  small-scale fading, we assume the Rayleigh fading channel model for the AP-user and IRS-user links while that for the AP-IRS link  will be  specified later depending on the scenarios.
  Without loss of generality, we assume that all IUs/EUs have the same SINR/RF receive power target, i.e., $\gamma_i=\gamma_0$, $\forall i\in \mathcal{K_I}$ and ${E}_{j}= {E}_0$, $\forall j\in \mathcal{K_E}$. For Algorithm 1, the phase shifts of all elements are initialized by $\theta_n=0, \forall n$,  $\{x_{i,k}, s_{j,i},t_{j,m}\}$, $i, k\in \K_{\II}$,  $j, m\in \K_{\E}$ are initialized randomly following $\mathcal{CN}(0,1)$, and the penalty coefficient is  initialized by $\rho=1000$. Other system parameters are set as follows unless specified otherwise  later: $c=0.9$, $\epsilon_1=10^{-4}$, $\epsilon_2=\epsilon_3=10^{-7}$, $d_{y1}=8$ m, $d_{y2}=100$ m, $r_I=r_E=2.5$ m, and $d_x=3.5$ m.  

\vspace{-0.4cm}
\begin{figure}[ht]
\centering
\subfigure[Constraint violation, $\xi$]{\includegraphics[width=0.49\textwidth]{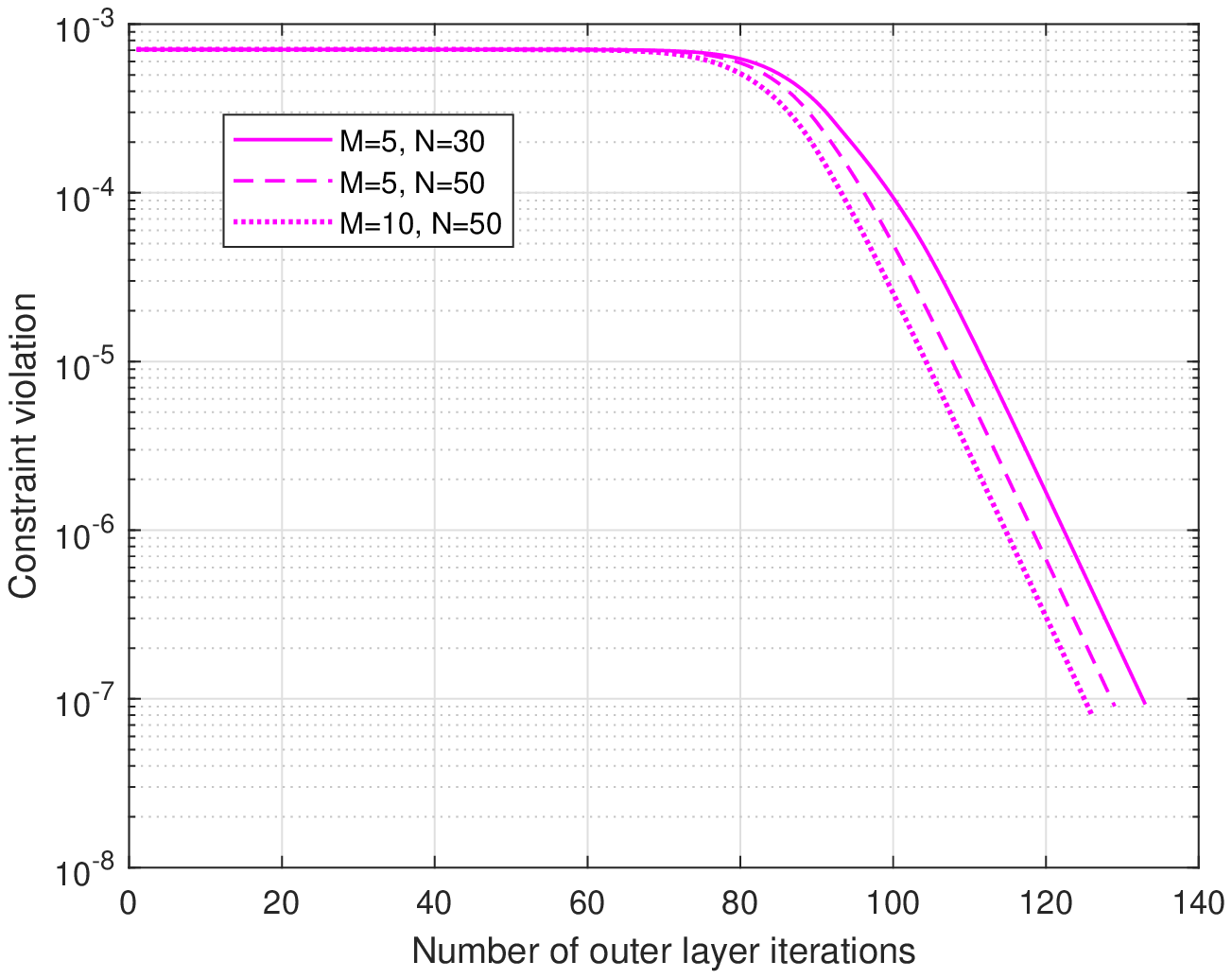} \label{simulation:constraint:violation} } 
\subfigure[Objective value]{\includegraphics[width=0.49\textwidth]{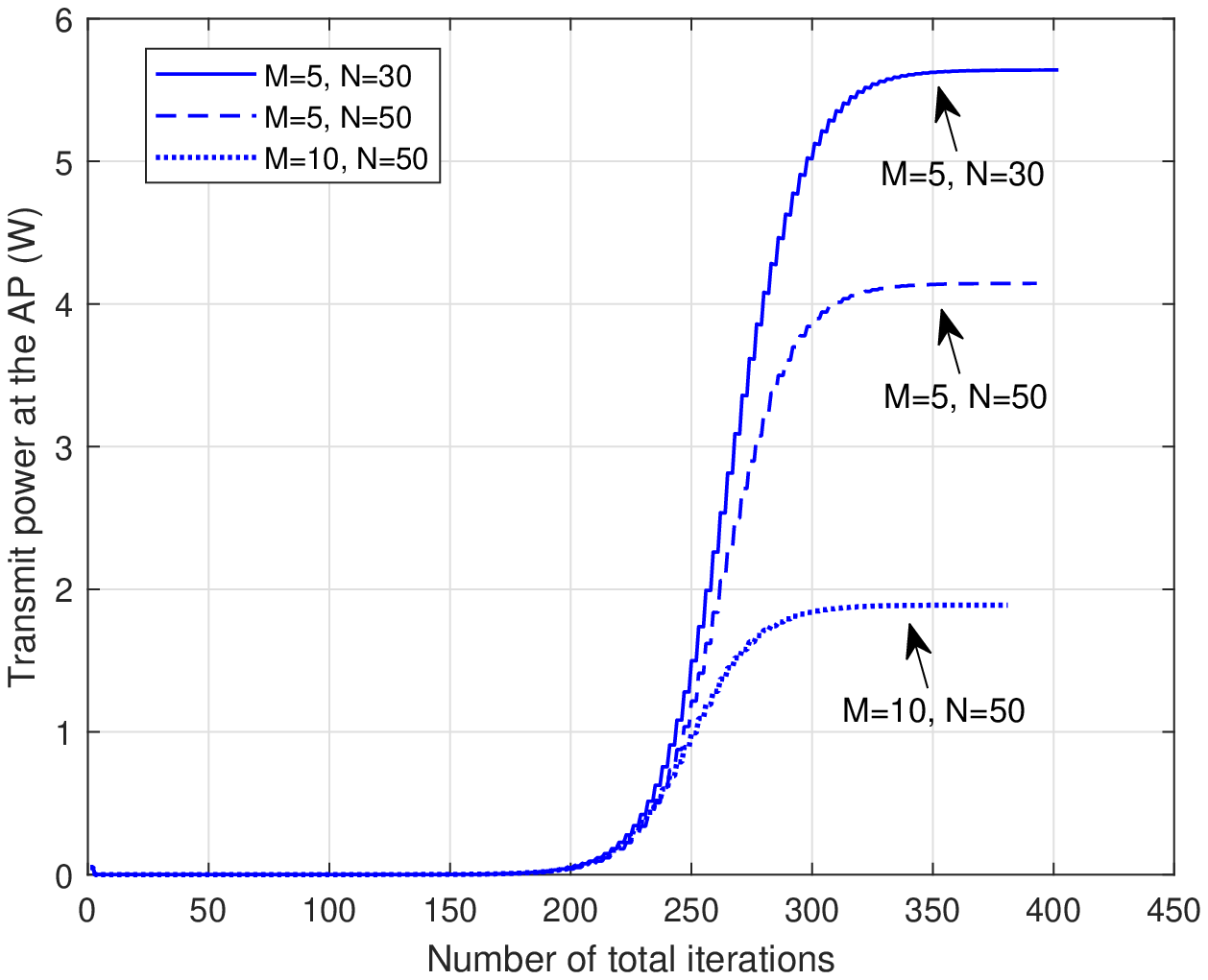}\label{simulation:obj:value}}
\caption{Convergence behaviour of Algorithm 1.   } \label{simulation:convergence}\vspace{-1cm}
\end{figure}
\subsection{Convergence of Algorithm 1}

Consider that $K_E=4$ and $K_I= 4$ and they are randomly located within the EU and IU clusters in front of IRS-1 and IRS-2, respectively, with ${E}_{0} =5$ microwatt ($\mu$W) and $\gamma_0=20$ dB.
In Fig. \ref{simulation:convergence}, we show the constraint violation for the newly introduced equality constraints, i.e., \eqref{constraint:violation}, and also the convergence of the proposed Algorithm 1 in Section III.
From Fig. \ref{simulation:constraint:violation}, it is observed that as the number of outer layer iterations increases, the equality constraints in (P2)  can  be eventually  satisfied within the  predefined accuracy (i.e., $10^{-7}$), which suggests that a solution satisfying all the user QoS constraints in  (P1) is obtained by Algorithm 1. Furthermore,  from Fig. \ref{simulation:obj:value}, one can observe that under different setups, the transmit power required at the AP  converges quickly. Note that the transmit power required at the AP  increases as the number of outer layer iterations increases (or equivalently $\rho$ decreases). This is expected since $\rho$ is initialized by a sufficiently large value, which results in sufficiently small transmit power in (P3) and a smaller $\rho$ corresponds to a larger penalty for the violation of equality constraints, which generally requires larger transmit power to minimize the penalty term.


%

\subsection{IRS-aided WPT}

We first study a special case of SWIPT, i.e., the IRS-aided WPT,  where there exist only EUs randomly located in the EU cluster.
   For comparison, we  consider three  benchmark schemes: 1) Alternating optimization where the transmit precoders and phase shifts are optimized alternately as in \cite{cui2019secure,guan2019intelligent,chen2019intelligent,dongfang2019,fu2019intelligent,huangachievable,yang2019intelligent,huang2018energy,wu2019weighted,pan2019intelligent}; 2)  IRS with fixed phase shifts, i.e., $\theta_n=0$, $\forall n \in \mathcal{N}$; and 3) Without IRS. The transmit precoders of 2) and 3) are  optimally obtained by applying semidefinite relaxation (SDR) as in \cite{xu2014multiuser}. To draw useful insight on the IRS deployment for WPT, we consider two cases of $\F$, i.e., $\F$ with all elements being 1 or $\F$ with all elements following  $\mathcal{CN}(0,1)$  independently, which correspond to deploying the IRS in an LoS-dominated and a rich-scattering environment with Rayleigh fading in practice, respectively.

 \begin{figure}[ht]
\centering
\includegraphics[width=0.55\textwidth]{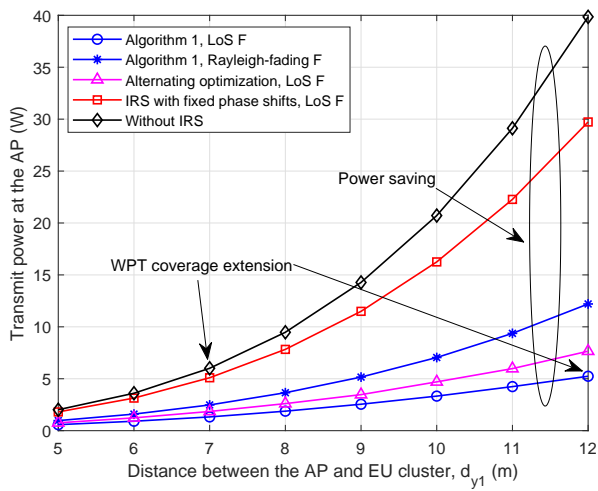}\vspace{-0.3cm}
\caption{Transmit power at the AP versus the AP-EU distance.  } \label{simulation:distance}\vspace{-0.5cm}  
\end{figure}

\subsubsection{AP Transmit Power versus AP-EU Distance}
To show the effectiveness of IRS for compensating the distance-dependent  path loss, we consider the circle center of the EU cluster moves along a line that is in parallel to $y$-axis shown in Fig. \ref{simulation:setup:multiIRS:SWIPT}, where  IRS-1 also moves accordingly to keep the relative distance with the EU cluster center unchanged. By varying the value of $d_{y1}$, we examine in Fig. \ref{simulation:distance} the  transmit power required at the AP  with $M=8$, $K_E=10$, $N_0=40$, and ${E}_{0} =5$ $\mu$W. From Fig. \ref{simulation:distance}, it is observed that without IRS, the transmit power required at the AP increases drastically as  EUs move far away from the AP. This thus fundamentally limits the operating range of  WPT due to practical restrictions (e.g., radio regulation) on the peak  transmit power. In contrast, by deploying the IRS in the proximity of EUs, the increase of the transmit power over distance is significantly alleviated.   In other words, for the same transmit power at the AP, the WPT operating range can be extended without compromising the RF receive power target at EUs.  For example, for the same transmit power about 6 W, EUs with the distance beyond 7 m from the AP cannot meet the energy harvesting constraint in the case without IRS, whereas by deploying the IRS in LoS with the AP,  it is even feasible for EUs with the distance of 12 m from the AP. This is expected since the large aperture and beamforming gain  of the IRS help boost  the signal power  significantly in its vicinity. Furthermore, one can observe that despite having the same average path loss between the AP and the IRS for both  LoS $\F$ and Rayleigh-fading   $\F$, the transmit power required in the former case is much lower than that in the latter case. This is expected since the rank-deficient LoS channel in the AP-IRS link introduces stronger correlation among the effective channels of EUs than its Rayleigh-fading  counterpart, thus rendering WPT more efficient.    Finally, it is observed that the proposed Algorithm 1 outperforms the alternating optimization scheme as well as the scheme with fixed phase shifts at the IRS.

\subsubsection{AP Transmit Power versus Number of EUs}

 \begin{figure}[ht]
\centering
\includegraphics[width=0.55\textwidth]{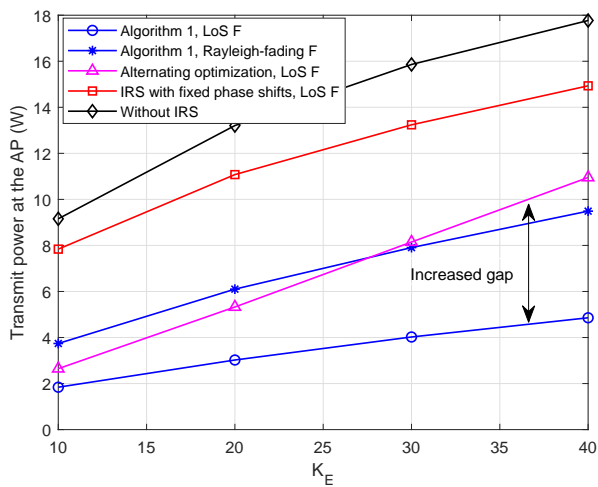}\vspace{-0.3cm}
\caption{Transmit power at the AP versus the number of EUs. } \label{simulation:WPT:EUs}\vspace{-0.5cm}
\end{figure}

To further investigate the IRS's effect on WPT performance as well as  the advantage of the proposed algorithm over the alternating optimization scheme, we plot in Fig. \ref{simulation:WPT:EUs} the transmit power required at the AP versus the number of EUs under the same setup as in Fig. \ref{simulation:distance} with $d_{y1}=8$ m. It is observed that as $K_E$ increases, the performance gap between the proposed algorithm and other benchmark schemes becomes larger. In particular, the alternating optimization scheme suffers considerable performance loss when the number of EUs (i.e., QoS constraints) becomes large. Moreover, besides significant transmit power saving, the number of energy beams required at the AP  is also reduced by deploying the IRS. To show this, we count in Table  \ref{table:energybeam}  the number of energy beams (denoted by $d_E$) required for the cases without IRS versus  with IRS (including both LoS $\F$ and Rayleigh-fading  $\F$), respectively,  for a total number of 500 channel realizations. From Table \ref{table:energybeam}, one can observe that without IRS, sending a single energy beam from the AP is suboptimal in general, especially for large $K_E$. However, by deploying the IRS in the case of  a rich scattering AP-IRS channel (Rayleigh-fading $\F$), the number of energy beams required is generally reduced, thanks to the higher  channel correlation induced by the additional  phase shifts at the IRS.  Furthermore, in the case of deploying the IRS in LoS with the AP  (i.e., LoS $\F$), one can observe that only one energy  beam is needed for all the considered cases, due to the strong channel correlation induced by the AP-IRS rank-one channel.  Based on the above, it is inferred  that even for the rich-scattering (Rayleigh-fading) channel between the AP and  IRS, more reflecting elements help  reduce the number of energy beams required since more degrees of freedom at the IRS can be leveraged to reconfigure and align the effective channels among EUs. As such, the deployment of IRS in WPT systems not only effectively reduces  the transmit power but also  simplifies the transmitter design by reducing the number of energy beams required at the AP.



\begin{table}[ht]
\caption{Results on the number of energy beams required: with (w/) IRS versus without (w/o)  IRS}\label{table:energybeam}\vspace{-0.4cm}
\centering
\newcommand{\tabincell}[2]
\small
\begin{tabular}{|l|l|l|l|l|l|l|l|l|l|}
\hline
\multicolumn{2}{|l|}{\multirow{2}{*}{}} & w/ IRS (LoS $\F$) & \multicolumn{3}{l|}{w/ IRS (Rayleigh-fading  $\F$)} & \multicolumn{4}{l|}{w/o IRS}          \\  \cline{3-10}
\multicolumn{2}{|l|}{}                  & $d_E=1$          & $d_E=1~\,\,~$                & $d_E=2$          & $d_E=3$         & $d_E=1$ & $d_E=2$ & $d_E=3$ & $d_E=4$ \\ \hline
\multicolumn{2}{|l|}{$K_E=10$}              & 500               & 499                   & 1                   & 0     & 73    & 426      & 1       & 0       \\ \hline
\multicolumn{2}{|l|}{$K_E=30$}              & 500               & 205                   & 294                   &1     & 0    & 37       & 439      & 24       \\ \hline
\multicolumn{2}{|l|}{$K_E=40$}              & 500               & 62                   & 415                  & 23      & 0   & 0       & 352      & 148      \\ \hline
\end{tabular}
\end{table}\vspace{-0.7cm}

\subsection{IRS-aided SWIPT}
Next, we consider the general case with both EUs and IUs coexisting in an IRS-aided SWIPT system.  First, by  assuming that only IRS-1 is deployed in Fig. \ref{simulation:setup:multiIRS:SWIPT},  we compare Algorithm 1 with other benchmark schemes  in subsection 1), and then study the impact of IRS on the SWIPT system in subsections 2) and 3). Finally,  the performance achieved by deploying both IRS-1 and IRS-2 is studied in subsection 4).
\subsubsection{AP Transmit Power versus Number of Reflecting Elements}
 In Fig. \ref{simulation:SWIPT:N}, we plot the transmit power required at the AP  versus the number of IRS reflecting elements  with $M=10$, $K_I=2$, $K_E=8$, $\gamma_0=20$ dB, and ${E}_0 = 10$ $\mu$W. The IUs and EUs are randomly located in their respective clusters (but without IRS-2 in Fig. \ref{simulation:setup:multiIRS:SWIPT}).  Besides the case without IRS, we also consider the case with discrete phase shifts at the IRS as well as a separate information-energy beam design for comparison.
Due to the large distance between IRS-1 and IUs,  IRS-reflected signals are negligible at IUs and thus can be ignored. As such, for the separate beam design,  the information beams are  first designed to minimize the transmit power required for satisfying all the SINR constraints by solving the following problem
\begin{align}  \label{MISO:MMSE}
&~\min_{\{\bm{w}_i\} } ~ \sum_{i\in \K_{\II}}\|\bm{w}_i\|^2  \\
\mathrm{s.t.} 
&~~\text{SINR}_i = \frac{|{\bm{h}}^H_{d,i}\bm{w}_i |^2}{\sum\limits_{ k\neq i, k\in \K_{\II} }|{\bm{h}}^H_{d,i}\bm{w}_k |^2   + \sigma^2_i} \geq \gamma_i, \forall i \in \K_{\II}.
\end{align}
Then the energy beams and phase shifts are jointly optimized to minimize the transmit power subject to the energy harvesting constraints as well as the constraint of no interference to all IUs (provided $K_I \le M-1$), i.e., \eqref{P1:EHconstrn} and $\bm{h}^H_{d,i}{\mv v}_j ={0},\forall  i \in \mathcal{K_I}, \forall j\in \mathcal{K_{E}}$.  This problem  can be solved similarly by Algorithm 1 in Section III.

 \begin{figure}[ht]
\centering
\includegraphics[width=0.55\textwidth]{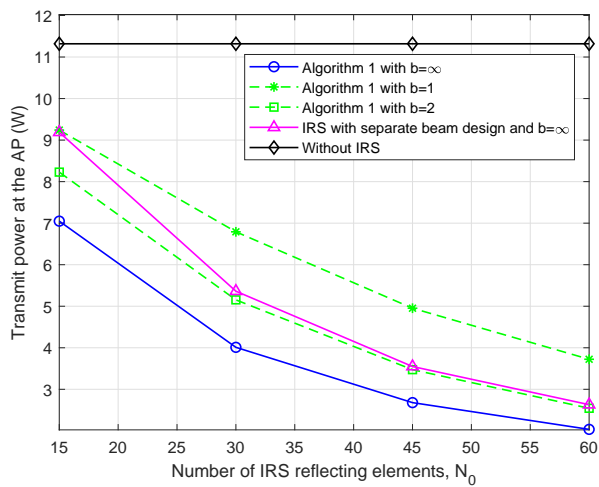}\vspace{-0.1cm}
\caption{Transmit power at the AP versus the number of IRS elements.} \label{simulation:SWIPT:N}\vspace{-0.4cm}
\end{figure}
From Fig. \ref{simulation:SWIPT:N}, it is first observed that  the use of discrete phase shifts at the IRS incurs performance loss as compared to the ideal case with continuous phase shifts, i.e., $b=\infty$, due to the misalignment of multiple reflected  signals.  However,  it still significantly outperforms the case without IRS, which shows the advantage of deploying IRS even with low-cost coarse phase-shifters  in practice. Furthermore, one can observe that the separate information-energy beam design suffers considerable performance loss as compared to the proposed joint design by Algorithm 1. It is worth pointing out that the performance gap between these two designs decreases as $N_0$ increases. This is expected since when the IRS's aperture becomes larger, the energy leakage from information beams  to IRS  also becomes more pronounced, and thus the transmit precoders of IUs in the joint design can be adjusted  to better serve IUs similarly  as the solution to problem \eqref{MISO:MMSE}.

\subsubsection{AP Transmit Power versus RF Receive Power Target at EUs}
 To further unveil the effect of IRS on the transmit precoder design at the AP, we consider the   benchmark schemes without sending dedicated energy beams, i.e., sending information beams only, for both cases with and without  IRS. These two optimization problems can be similarly solved by Algorithm 1. In Fig. \ref{simulation:SWIPT:power:vs:EHtarget}, we compare the  transmit power versus the RF receive power target at EUs with $K_I=1$, $K_E=16$, and $N_0=50$ (other parameters are set to be the same as in Fig. \ref{simulation:SWIPT:N}). Since $K_I=1$,  for the case without energy beam,  only one (information) beam is sent from the AP.  In addition, eight EUs are assumed to be located in the EU cluster and the other eight EUs are assumed to be evenly distributed on the left semicircle centered at the AP on the same plane with the radius of 8 m.   From Fig. \ref{simulation:SWIPT:power:vs:EHtarget}, it is observed that  for both the cases with and without energy beams, the transmit power reduction achieved by deploying the IRS  becomes more evident as the RF receive power target of EUs  increases.
  This is intuitive since when ${E}_{0}$ is very small, the energy leakage from the information beam  is already sufficient in most cases and deploying the IRS around EUs only brings marginal performance gain, whereas for high ${E}_{0}$, the usefulness of deploying the IRS for compensating the path loss is more evident.
More importantly, one can observe that by deploying the IRS, the transmit power saved by sending dedicated energy beams is largely reduced, which suggests that in most cases, sending the information beam only at the AP is already satisfactory. This is because without IRS, the AP sending one information beam only has to steer its beam direction to strike a balance between the IU and multiple EUs, which is inefficient for large $K_E$ with high ${E}_{0}$ (which is consistent with Table \ref{table:energybeam} where $d_E=1$ is generally suboptimal even for $K_E=10$ in the case without IRS). In contrast,  allowing the AP to send dedicated energy beams  helps resolve the above  issue and thus  achieves the so-called energy beamforming gain effectively.  However, by deploying IRS-1 around EUs, the effective channel power gains of EUs are significantly improved, which thus reduces the transmit power allocated for dedicated energy beams and in most cases, even the energy leakage from the properly steered information beam with optimized IRS  phase shifts is sufficient to meet the energy harvesting constraints for EUs. As such, the necessity   of sending dedicated energy beams is weakened in IRS-aided SWIPT systems, which implies that the deployment of IRS can  potentially simplify the transmit beamforming design at the AP if no energy beams are sent  as well as the receiver of  IUs since no energy signal cancelation component is needed.
 \begin{figure}[!t]
\centering
\includegraphics[width=0.55\textwidth]{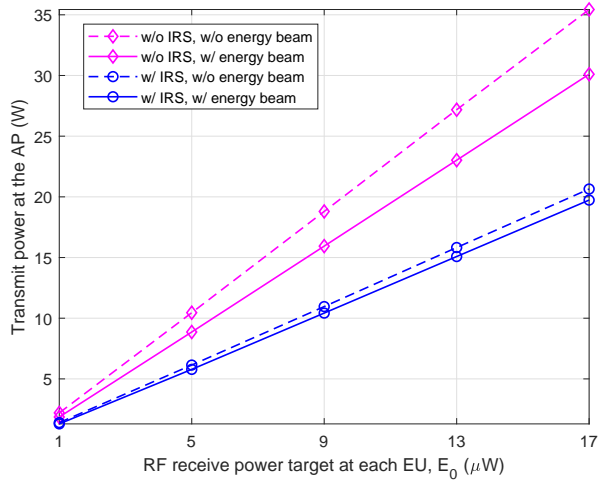}\vspace{-0.3cm}
\caption{Transmit power at the AP versus the  RF receive power target of EUs with $K_I=1$. } \label{simulation:SWIPT:power:vs:EHtarget}\vspace{-0.6cm}
\end{figure}

 \begin{figure}[ht]
\centering
\includegraphics[width=0.55\textwidth]{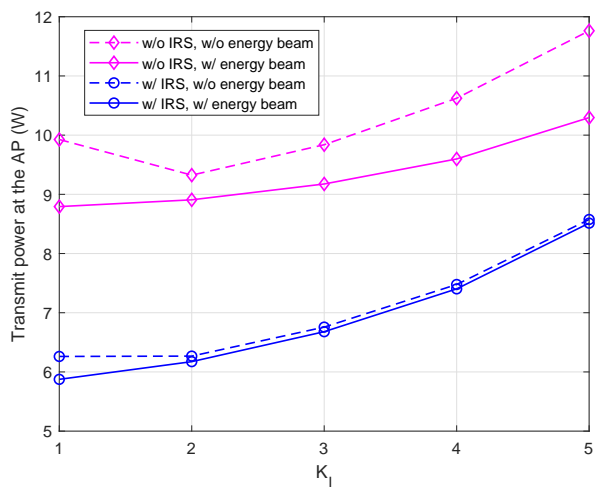}\vspace{-0.3cm}
\caption{Transmit power at the AP versus the number of IUs.} \label{simulation:SWIPT:Ki}\vspace{-0.4cm}
\end{figure}

\subsubsection{AP Transmit Power versus Number of IUs}
Based on the same setup in Fig. \ref{simulation:SWIPT:power:vs:EHtarget}, we gradually  increase the number of IUs (which are randomly located in the IU cluster) to study its effect on the transmit power required at the AP with ${E}_0 = 5$ $\mu$W for EUs, shown in Fig. \ref{simulation:SWIPT:Ki}. First, one can  observe that with the deployment of the IRS, sending dedicated energy beams only brings negligible transmit power reduction, especially when $K_I$ is large, which is consistent with our discussion for Fig. \ref{simulation:SWIPT:power:vs:EHtarget}.  Besides, it is  observed from Fig. \ref{simulation:SWIPT:Ki} that without the IRS,  increasing $K_I$ from 1 to 2 even reduces the transmit power required for the case without energy beam, in contrast to the case with energy beam where adding IUs generally requires higher transmit power at the AP. This is because with $K_I=2$, the AP can send two information beams and thus has higher flexibility to balance the RF receive power among EUs,  resulting in higher beamforming gain and thus  lower transmit power as compared to $K_I=1$.  However, this phenomenon is not observed  when the IRS is deployed around EUs due to the additional IRS  beamforming gain.
Last, one can observe from Fig. \ref{simulation:SWIPT:Ki} that as $K_I$ increases, the performance gain achieved by deploying the IRS around EUs decreases, especially for the case with energy beamforming. This is expected since IUs gradually become the performance bottleneck of the SWIPT system due to the more  severe multiuser interference and deploying an IRS around EUs alone is no more  sufficient. This implies that in practice, additional IRS may need to be deployed around IUs if the number of IUs is large and/or their SINR target becomes high, which will be evaluated in the next subsection.

 \begin{figure}[ht]
\centering
\includegraphics[width=0.55\textwidth]{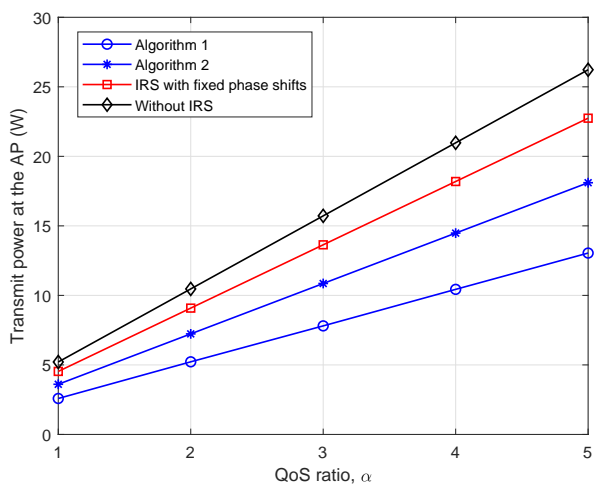}\vspace{-0.3cm}
\caption{Transmit power at the AP versus the QoS ratio.  } \label{simulation:multipleIRS:EH}\vspace{-0.4cm}
\end{figure}

\subsubsection{AP Transmit Power versus QoS Ratio}Finally, we study the multi-IRS aided SWIPT system shown in Fig. \ref{simulation:setup:multiIRS:SWIPT} where IRS-2 is further deployed to improve the performance of IUs. Motivated by \cite{JR:wu2018IRS}, we assume  Rayleigh fading channel model between the AP and IRS-2 to reap the spatial multiplexing gain for IUs.  Besides, we consider $K_I=6$ and  $K_E=8$ with two IUs (far from IRS-2) and two EUs (far from IRS-1) located in $(d_x-100 ,0,0)$, $(d_x+100 ,0,0)$, $(d_x-8 ,0,0)$, and $(d_x+8 ,0,0)$, respectively. The rest  IUs and EUs are randomly distributed in their respective clusters.
In Fig. \ref{simulation:multipleIRS:EH}, we compare the transmit power required by Algorithm 1 and  the low-complexity Algorithm 2  versus the QoS ratio, denoted by $\alpha \geq 1$, with $M=10$ and $N_0=40$. Specifically, we set  $\gamma_0 = \alpha {\bar \gamma}_0$ (in linear scale) and $E_0= \alpha {\bar E}_0$ with  $ {\bar \gamma}_0= 10$ and $ {\bar E}_0=4$ $\mu$W. By increasing the value of $\alpha$, both the SINR and RF receive power requirements of IUs and EUs increase. From Fig. \ref{simulation:multipleIRS:EH}, it is observed that although the low-complexity Algorithm 2 suffers some performance loss as compared to the penalty-based Algorithm 1, it still significantly outperforms the case without IRS as well as the case with fixed phase shifts at IRSs. This further demonstrates the effectiveness of IRSs in enhancing the performance of SWIPT systems with both IU and EU hot spots.

\section{Conclusions}

In this paper, we investigate a new QoS-constrained  beamforming optimization problem for IRS-aided SWIPT.  Specifically, the active transmit precoders  at the AP and the passive reflect phase shifts at multiple  IRSs are jointly optimized to minimize the transmit power at the AP subject to both the SINR constraints at  IUs and energy harvesting constraints at EUs. We propose two algorithms to achieve a balance between the system performance and the computational complexity. In particular, the proposed penalty-based algorithm is shown to yield the best performance as compared to other existing and benchmark schemes. Furthermore, it is  able to handle  the practical case with discrete phase shifts as well.  Simulation results validate the effectiveness of IRS for WPT range extension and transmit power saving under SWIPT constraints. Besides, it is found that the deployment of IRS can effectively reduce the number of energy beams required for WPT/SWIPT systems.


\end{document}